\documentclass[11pt]{article}
\usepackage{amsmath,amsthm,amsfonts,amssymb,amscd}
\usepackage[dvips]{graphicx,psfrag}
\usepackage{latexsym}
\newtheorem{theorem}{Theorem}
\newtheorem{proposition}{Proposition}
\newtheorem*{proposition*}{Proposition}

\newcommand{\ap} {\alpha}
\newcommand{\bt} {\beta}

\newcommand{\Dt} {\Delta}
\newcommand{\ep} {\epsilon}

\def\R{{\mathbb R}}

\newcommand{\nd} {\noindent}

\begin{document}
\begin{titlepage}

\begin{center}

{\LARGE \bf Metastable periodic patterns in }
{\LARGE  \bf singularly perturbed state dependent delayed equations}

\vskip  .3truecm

{Xavier Pellegrin$^{c,}$\footnote{Corresponding author : pellegrin.xavier@ijm.univ-paris-diderot.fr, +33 (0)157278056.},
 C. Grotta-Ragazzo$^a$,  C.P. Malta$^b$, and  K. Pakdaman$^c$}
\\
\end{center}
\vskip 2mm
{\small
\begin{center}
a) Instituto de Matem\'atica e Estat\'\i stica\\
Universidade de S\~ao Paulo\\
05508-090, S\~ao Paulo, SP, BRASIL\\
\vskip 2mm
b) Instituto de F\'\i sica\\
Universidade de S\~ao Paulo\\
R. do Mat\~ao, Travessa R, 187\\
05508-090 S\~ao Paulo, BRASIL\\
\vskip 2mm
c) Univ Paris Diderot, Sorbonne Paris Cit\'e,\\ 
IJM, UMR 7592 CNRS,\\ 
F-75205 Paris Cedex 13, France\\
\vskip 2mm
\end{center}
}
\vskip .3truecm

\begin{abstract}
  We consider the scalar delayed differential equation $\ep\dot
  x(t)=-x(t)+f(x(t-r))$, where $\ep>0$, $r=r(x,\ep)$ and $f$
  represents either a positive feedback $df/dx>0$ or a negative
  feedback $df/dx<0$.  
When the delay is a constant, i.e.
  $r(x,\ep)=1$, this equation admits metastable rapidly oscillating
  solutions that are transients whose duration is
  of order $\exp(c/\ep)$, for some $c>0$. 
In this paper we
  investigate whether this metastable behavior persists when the delay
  $r(x,\ep)$ depends non trivially on the state variable $x$.  
Our
  conclusion is that for negative feedback, the persistence of the
  metastable behavior depends only on the way $r(x,\ep)$ depends on
  $\ep$ and not on the feedback $f$. 
In contrast, for positive
  feedback, for metastable solutions to exist it is further required
  that the feedback $f$ is an odd function and the delay $r(x,\ep)$ is
  an even function. 
Our analysis hinges upon the introduction of state
  dependent transtion layer equations that describe the profiles of
  the transient oscillations.  
One novel result is that state
  dependent delays may lead to metastable dynamics in equations that
  cannot support such regimes when the delay is constant.
\end{abstract}

\vskip .3truecm

\nd {\bf Key words:} metastability, state-dependent delayed differential
equation,
singular perturbation.

\vfill \hrule \smallskip \nd {\small CGR is partially supported by
  CNPq (Brazil) grant 305089/2009-9 , CPM is partially supported by
  CNPq (Brazil) grant 311022/2009-0. }

 \end{titlepage}

\section{Introduction}
\label{intro}

\par State dependent delay differential equations (DDEs) of the form
\begin{equation}
\ep {dx\over dt}(t)=\ep \dot x(t)=-x(t)+f(x(t-r)), \quad {\rm where}\quad
r=r(x,\ep)= r_0 + \tilde r(x,\ep) \; ,
\label{eqmain}
\end{equation}
appear as models in economics, physics and biology
\cite{JB-MM_89,MM_89,AL-JM_89,AL-JM_89B,OA-ES-AF_99,CF-MM_09}.  Many
studies have dealt with the asymptotic dynamics of these equations,
establishing existence, stability and profiles of the so called slowly
periodic oscillations when $f$ represents a negative feedback (see,
for instance,
\cite{alt,arino-hbid,krisztin-arino,smith,MPNP,walther,mallet1,MPN4,MPN5,Nussbaum_2003}).
There are also results on existence and stability of periodic
solutions, and on the convergence of most solutions to equilibria in
the case of monotone positive feedback \cite{kris,polner}. For a
review on DDEs with state dependent delays see \cite{Hartung_2006}.
In contrast with these previous studies, in addition to the long term
behavior of solutions, we are also interested in the transient
dynamics of Eq.  (\ref{eqmain}): we explore conditions under which the
system displays metastable transient oscillations prior to convergence
to its asymptotic state.  Typically, in such a regime, the system
engages in seemingly stable and sustained oscillations for a time of
order $\exp(c/\epsilon)$ for some constant $c$, before eventually
converging to a stable equilibrium or a stable periodic solution.

\par It is well known that delayed feedback can lead to oscillations in a
system that would otherwise remain at rest. Much focus has been on
delay induced sustained undamped oscillations, such as periodic ones.
However, the delay can also produce transient oscillations, that end
up vanishing with time. While these had been reported in several
studies (see for instance \cite{shar,politi2,politi,pre_ring} and
references therein), for a long time they had not been considered as a topic of
investigation on their own. In \cite{pre_ring}, the term DITO was
coined for this phenomenon both as an acronym for ``delay induced
transient oscillation'' and as a means to emphasize its repetitive
nature.  In the same paper, it was argued that long lasting DITOs
could strongly alter the dynamics of the system and, for instance,
interfere with information retrieval in neural networks. This work
together with the ones cited above initiated intrinsic interest in the
phenomenon. Since these early studies, DITOs have been reported in
other systems, such as for instance \cite{Milton_2010,Milton_2011}.
In support of claims in \cite{shar} and numerical
explorations in \cite{pre_ring}, through analytical treatment of a
specific system of DDEs with piecewise constant feedback, it was
shown, in \cite{pre_expo}, that some DITOs could persist for durations
of the order $\exp(c/\epsilon)$. Such long transients would outlast
any observation window when $\epsilon$ is small enough and would be
undistinguishable from (nearly) periodic solutions. In analogy with
long lasting transient oscillations in partial differential equations,
DITOs with such exponential lifetimes are referred to as metastable
\cite{pre}.

\par Metastable oscillations have been thoroughly studied in many systems
such as scalar partial differential equations \cite{cp,fh}.  However
they have received far less attention in DDEs. Here we examine their
occurrence in DDE (\ref{eqmain}).  Our previous works characterized the
conditions for metastability in equations with constant delays for
both positive and negative feedbacks \cite{pre,jdde}. It further
revealed a remarkable difference in metastability between the positive
and negative feedback equations, that had escaped earlier reports and
explorations.  Indeed, a positive feedback DDE presents metastability
only when the feedback $f$ satisfies a special symmetry property that
holds, for instance, when $f$ is an odd function.  This condition is
absent in DDEs with negative feedback 
(see also \cite{niz1,niz2} for related results).
One consequence of this result is that not all DITOs are metastable. 

\par All these previous works deal with
constant delays. As far as we
know, the present paper is the first one to investigate and show the existence
of metastability for state-dependent delay equations.
Our goal in this paper is to study the transient dynamics of Eq.
(\ref{eqmain}) when the parameter $\ep>0$ is small or, equivalently,
$r_0$ is large. In previous singular perturbation analyzes of DDEs
with state dependent delay \cite{MPN1,MPN2,MPN3}, the small parameter
$\ep$ appears only in the left hand side of the equation. Here, the
novelty resides in the fact that $\ep$ appears in Eq. (\ref{eqmain})
in two different places: multiplying the time derivative of $x$, which
characterizes a singular perturbation problem, and as an argument of
$\tilde r$.  One major contribution of this work is to highlight the fact
that the dynamics of Eq.~(\ref{eqmain}) strongly depend on the way $\tilde r$
depends on $\ep$.

\medskip
\par This paper is organized as follows. In section 2 we present and
discuss some previous results on the existence of metastable solutions
in the context of DDEs with constant delays.  These previous results
indicate that there are two essential ingredients for the occurrence of
metastability. The first is the existence of rapidly oscillating
periodic solutions (in a sense to be defined in section
\ref{metastability}) as the singular parameter $\ep$ tends to zero.
The second is that the shapes of these oscillations tend to
stereotyped ``square-wave'' profiles as $\epsilon \rightarrow 0$.
These profiles are heteroclinic solutions of certain transition layer
equations.  A key consequence of this is that metastability can be
characterized through the analysis of these transition layer
equations.  In sections \ref{sec:periodic_sol} and \ref{sec_tle} we
examine the same two ingredients namely the existence of rapidly
oscillating periodic solutions and transition layer equations for
state dependent delay equations.  More precisely, in section
\ref{sec:periodic_sol}, using a Hopf bifurcation theorem, we examine
the existence of rapidly oscillating periodic solutions for
Eq. (\ref{eqmain}), regardless of the choice of $R$ and its scaling
with the small parameter $\ep$ (details of the computations are
presented in the appendix A).  Then we link the amplitude of rapid
periodic oscillations to the possible existence or absence of
metastability in some of the state-dependent DDEs under study.  In
section \ref{sec_tle} we introduce transition layer equations to
describe metastability (more details are given in Appendix B),
similarly to what has been done for DDEs with constant delay
\cite{chowmp,m-pn,jdde}.  We show that symmetry conditions on
solutions of these equations can be used to characterize the condition
on the feedback function $f$, and the delay function $r$, required for
the existence of metastable oscillating solutions.  One novel
observation, specific to equations with state dependent delay, is that
the existence of metastable solutions may depend on the behavior of the
function $r$ when $\epsilon \rightarrow 0$.
In section \ref{section_num}, we present systematic numerical
explorations of the transient dynamics of Eq. (\ref{eqmain}).  Section
\ref{sec_conclusion} contains a discussion and a summary of the
results.

\medskip
\par Prior to section \ref{metastability}, we go
over some assumptions and notations.  In most applications the
functions $f$ and $r$ are differentiable in $x$, here this is always assumed.
Furthermore, in accordance with our previous results \cite{pre,jdde},
we assume that $f$ is either a monotonic increasing (positive
feedback) or a monotonic decreasing (negative feedback) function
satisfying the following hypotheses.
\begin{itemize}
\item[] {\bf Positive Feedback} \\
   $f'(x) \ge 0$,  $f(0)=0, f'(0)>1$, and
  there exist $a>0$, $b>0$, such that $f(-a)=-a$, $f(b)=b$,
  $0<f'(-a)<1$, $0<f'(b)<1$, and
  $f(x)\neq x$ for $x \in (-a,0) \cup (0,b)$. 
\item[] {\bf Negative Feedback} \\
  $f'(x) \le 0$, $f(0)=0, f'(0)<-1$, and
  there exist $a>0$, $b>0$, such that $f(-a)= b$, $f(b)=-a$, $0 <
  f'(-a) f'(b) <1$; $|f(f(x))| > x$ for $ x \in (-a,0) \cup (0,b)$,
  and  $|f(f(x))| < x$ for $ x \in (-\infty,-a) \cup (b,\infty)$. 
\end{itemize}
In the following, a positive or negative feedback $f$ will be called
symmetric if it is an odd function.  Under these hypotheses, both the
dynamics of the map $f:\mathbb{R} \to \mathbb{R}$, and the DDE Eq.
(\ref{eqmain}), with constant delay are well understood. In
particular, the map $f$ in positive feedback case has only three fixed
points: $x=0$, unstable, and $x= -a$, $x=b$, stable. For negative
feedback, the map has a single unstable fixed point $x=0$, and a
single attracting 2-periodic orbit $\{-a,b\}$, $R(x)=0$ corresponding
to constant delay.
Without loss of generality, one can use the time scaling $s =
\frac{t}{r_0}$, and consider Eq. (\ref{eqmain}) with $r_0= 1$.  In the
following we shall assume that the delay function is
\begin{equation}\label{eq_delay_eta}
r(x,\ep)=1+ \eta(\ep) R(x),
\end{equation}
where $\eta$ is a non negative and smooth function on $\epsilon \geq
0$.  In this way, different scalings between the delay $r$ and the
parameter $\epsilon$ can be analyzed by changing the function $\eta$.
Typically $R$ will be chosen as a smooth nonnegative function.

\section{Metastability in DDEs with constant delay}
\label{metastability}

Metastability is a concept that appears in several branches of physics
and mathematics. Schematically, a
state is called metastable if it is transient, that is, it is
eventually transformed into another one, but this transformation is on
such a slow time scale that it is not perceived in normal observation
windows. One of the first systems of relevance to the present work, in which
metastability had a full mathematical treatment, is that of scalar
parabolic equations \cite{cp,fh}. Metastable solutions of DDEs
(\ref{eqmain}) with constant delay $r(x,\ep)=1$ share a number of
features with those of the partial differential equations
\cite{pre,jdde}. In this section, we provide an overview of metastable
solutions in scalar DDEs with constant delay, 
that will serve as a basis for the
comparison and analysis of the case of state dependent delays.

Throughout the remainder of this section, we refer only to DDEs
(\ref{eqmain}) with constant delay set to one.  Roughly speaking,
metastable solutions of these equations are trajectories that evolve
close to unstable manifolds of unstable periodic orbits, 
and they seem to be periodic, while in fact they are endowed with a drift 
of order $\exp(c/\ep)$. 
In section \ref{subsec_rap} we describe the relationship between rapidly
oscillating periodic orbits and metastable oscillations, and in section
\ref{subsec_drift} we show how transition layer equations allow to
quantify the slow drift of such oscillations.

\subsection{Metastability and rapidly oscillating periodic orbits}
\label{subsec_rap}

To make the description more concrete, consider, for instance, the
positive feedback case $f(x)= \frac12 \arctan (5x)$.  For any value of
$\ep>0$, Eq (\ref{eqmain}) has exactly three equilibria, $x=-a$,
$x=0$, and $x=a$, where $a$ is the positive solution of $x=
\frac12 \arctan(5x)$.  For all values of $\ep>0$, the equilibria $x=-a$
and $x=+a$ are locally asymptotically stable and $x=0$ is unstable.
For all $\ep$ sufficiently large the unstable manifold of $x=0$ is
one-dimensional and as $\ep$ decreases, $x=0$ undergoes an infinite
number of Hopf bifurcations at $\ep_1>\ep_2>\ldots>0$ such that the
dimension of the unstable manifold of $x=0$ at $\ep\in
(\ep_{n},\ep_{n+1})$ becomes $2n+1$.  Let $x_{n}$ denote the periodic
solution that bifurcates from $x=0$ at $\ep=\ep_n$.  If $x_n(t)=0$
then $x_n(t+s)$ has exactly $2n$ zeroes for $s\in [-1,0]$.  The
branches of periodic orbits that appear at these successive Hopf
bifurcations can be extended up to $\ep=0$, 
the period converges to $r=r_0=1$ and the amplitude converges to $a >0$
\cite{arino_benkhalti}, and the oscillations tend to a
square-wave-like shape when $\epsilon \rightarrow 0$.  So, Eq.
(\ref{eqmain}) admits periodic solutions with a large number of zeroes
in a time interval of length one, provided that $\ep>0$ is
sufficiently small.  Given $\ep>0$ the global attractor of Eq.
(\ref{eqmain}) consists of its set of equilibria and periodic orbits,
and their finite dimensional unstable manifolds.  These solutions are
organized in a peculiar way: the global attractor admits a Morse
decomposition, with each Morse set containing a periodic orbit or an
equilibrium. The direction of the flow on this decomposition is such
that the number of zeroes of solutions is a non increasing function of
time \cite{arino_seguier,polner}.  A similar description of the
periodic orbits, their branches, their shape, and the organization of
the trajectories on the global attractor holds for Eq. (\ref{eqmain})
with negative feedback \cite{mallet1,m-pn}.

Now that we have depicted the long term dynamics of DDEs with positive
and negative feedback, we can describe the way metastable solutions
appear in Eq. (\ref{eqmain}).  For $\ep>0$ small, an initial condition
$\varphi:[-1,0]\to\R$ to Eq.  (\ref{eqmain}) with $2n$ zeroes gives
rise to a solution $x_t(\cdot,\varphi,\ep):[-1,0]\to\R$ that after a
time of order one is pointwise close to a function in the unstable
manifold of $x_{n}$, that has a square-wave-like shape close to that
of $x_n$. The dynamics of $x_n$ is
approximately periodically oscillatory, but the slow motion of the
solution along the unstable manifold of $x_n$ eventually annihilates a
pair of zeroes of the solution, 
which takes a time of order $\exp (c/\ep)$.  After the
annihilation of the two zeroes the solution drifts along the unstable
manifold of another periodic solution $x_{n-1}$, and this process repeats
itself until the solution eventually approaches one of the two stable
equilibria $x=-a$ or $x=a$ for positive feedback, or a slowly
oscillating periodic orbit for negative feedback.  
Metastability means here that 
the solutions have a (fast) oscillatory transient time 
that grows as $\exp (c/\ep)$ when $\ep\to 0$.

\subsection{Metastability and transition layers}
\label{subsec_drift}

Metastability not only depends on the existence of periodic solutions
(i.e. the qualitative geometry of the phase portrait), but also 
on the quantitative dynamics along their unstable manifolds. 
Metastability happens only when the rapidly oscillating
solutions are square-wave-like, in which case their jumps
have been understood as transition layer phenomena, and 
described and analyzed using transition layer equations.

Suppose that $f$ is a positive feedback and that $x(t)$ is an
oscillatory metastable solution of (\ref{eqmain}), with $r=1$, that
jumps from $x(t)\approx b>0$, for $t<0$, to $x(t)\approx -a<0$, for
$t>0$, with $x(0)=0$.  The approximately periodic behavior of the
metastable solution implies that $x(t)\approx x(t+1+\rho\ep)$.  When
$\epsilon$ is small, this and Eq. (\ref{eqmain}) imply
\[ \ep {dx\over dt}(t)=-x(t)+
f\{x(t- 1 )\}
\approx -x(t)+ f\{x(t+\rho\ep) \}.   \]
Rescaling $x(t)$ as
 $\phi^-(t)= x(\epsilon t)$, the above equation for $x(t)$ implies
that $\phi^-$ must  satisfy the transition layer equation 
\begin{equation}
\label{eqTLE--}
\dot  \phi^-(t)  = - \phi^-(t) + f( \phi^-( t + \rho^-  ) ),
\end{equation}
where $\rho^->0$ is an unknown constant and the function $\phi^-$ must
satisfy the boundary conditions $\lim_{t\to -\infty}\phi^-(t) = b$ and
$\lim_{t\to \infty}\phi^-(t) = -a$. 
The same ideas apply to a jump 
from $x(t)\approx b>0$ for $t<0$, to $x(t)\approx -a<0$ for $t>0$, 
and leads to the existence
of an increasing transition layer solution, i.e. a solution to
\begin{equation}\label{eqTLE++}
\dot  \phi^+(t)  = - \phi^+(t) + f( \phi^+( t + \rho^+  ) )
\; \; \; \; 
\end{equation}
where $\rho^+>0$ is an unknown constant and the function $\phi^+$ must
satisfy the boundary conditions $\lim_{t\to -\infty}\phi^+(t)=-a$ and
$\lim_{t\to \infty}\phi^+(t)=b$. 

The constants $\rho^+$ and $\rho^-$ are drift velocities of ascending
and descending sign-changes (zeroes) of an oscillatory solution. 
If an oscillatory solution $x(t)$
of Eq. (\ref{eqmain}) satisfies $x(t_0^+)=0$ and $x^\prime(t_0^+)>0$,
then one expects to find a time $t_1^+ \approx t_0^+ + 1 + \epsilon
\rho^+$ such that $x(t_1^+)=0$ and $x^\prime(t_1^+)>0$; if
$x(t_0^-)=0$ and $x^\prime(t_0^-)<0$, then one expects to find a time
$t_1^- \approx t_0^- + 1 + \epsilon \rho^-$ such that $x(t_1^-)=0$ and
$x^\prime(t_1^-)<0$. This suggest that the symmetry condition
$\rho^+=\rho^-$ should be associated with metastability of oscillating
solutions.  Indeed,  exponential duration of
oscillatory transients is proven in \cite{jdde} for scalar DDEs, with monotone
positive feedback, by an estimate of the type $ \vert t_1^\pm -
(t_0^\pm + 1 + \epsilon \rho^\pm) \vert \leq e^{c/\epsilon}$, so that
when $\rho^+=\rho^-$ oscillatory solutions are close to a $1 +
\epsilon \rho$ periodic solution up to an exponential order.

Finally, we discuss transition layer for negative feedback function.
Transient oscillations are square-wave-like and they have an
approximate period $T \approx 2 + \epsilon C$.  If an oscillatory
solution $x(t)$ of Eq. (\ref{eqmain}) satisfies $x(t_0)=0$ and
$x^\prime(t_0)>0$, then one expects to find a time $t_1 \approx t_0 +
1 + \epsilon \rho^+$ such that $x(t_1)=0$ and $x^\prime(t_1)<0$, and
then a time $t_2 \approx t_1 + 1 + \epsilon \rho^-$ such that
$x(t_2)=0$ and $x^\prime(t_2)>0$.  So the transition layer equations
for the increasing and decreasing transition layer solutions are
coupled:
\begin{equation}
\left \{
\begin{array}{ccc}
\dot \phi^+(t) &=& - \phi^+(t) + f( \phi^-[ t  + \rho^- ] ) \; , \\
\dot \phi^-(t) &=& - \phi^-(t) + f( \phi^+[ t  + \rho^+ ] ) \; , \\
\end{array}
\right .
\label{eqTLEneg_c}
\end{equation}
where $\phi^+$ is increasing and $\phi^-$ is decreasing on
$\mathbb{R}$, with $\lim_{t\to -\infty}\phi^+(t) = -a$, $\lim_{t\to
  \infty}\phi^+(t) = b$, $\lim_{t\to -\infty}\phi^-(t) = b$ and
$\lim_{t\to \infty}\phi^-(t) = -a$, with $\phi^+(0) = \phi^-(0) = 0$ and
$\rho^\pm$ are unknown real constants.
For negative feedbacks, given an oscillatory solution $x(t)$, an
(ascending) zero $x(t_0)=0$ with $x^\prime(t_0) >0$ gives rise to a
(descending) zero $x(t_1)=0$ with $x^\prime(t_1) >0$, and then to
another ascending zero $x(t_2)=0$ with $x^\prime(t_2) >0$, with $t_1
\approx t_0 + 1 + \epsilon \rho^+$ and $t_2 \approx t_1 + 1 + \epsilon
\rho^-$. So, in first order, the drift speeds of ascending and
descending zeros are identical and equal to $2 + \epsilon(\rho^- +
\rho^+)$.  Hence the symmetry condition, equivalent to the condition
$\rho^+=\rho^-$ in the positive feedback case, is here $\rho^+ + \rho
^- = \rho^- + \rho^+$, and it is obviously always true, irrespective
of $\rho^+$ and $\rho^-$ \cite{jdde}.

In summary, the analysis using transient layer equations reveals and
explains that metastability manifests itself in different ways in
positive and negative feedback systems. Indeed, positive feedback DDEs
present metastability only when a very special symmetry property of
the transition layer problems (\ref{eqTLE--}) and (\ref{eqTLE++}) is
satisfied, which holds for instance when $f$ is an odd function, while
there is no such restriction in the case of negative feedback
equations (\ref{eqTLEneg_c})~\cite{pre,jdde}.

\section{Existence and amplitude of periodic solutions}
\label{sec:periodic_sol}

\par Based on literature results and novel results that follow in
this section, we conjecture that, provided the delay satisfies a
number of classical technical assumptions
\cite{MPN1,MPN2,MPN3,krisztin-arino,Hartung_2006,eichmann,bartha_monotone,bartha_periodic},
the geometric organization of the phase portrait of DDEs with state
dependent delays and monotone feedback is similar to that of equations
with constant delays.  Namely our three conjectures are:
(i) Equation (\ref{eqmain}), with $\tilde{r} (x,0) = 0$, sustains
branches of periodic solutions, that appear at the $0$ equilibrium by
successive Hopf bifurcations and exist until $\epsilon \rightarrow 0$,
with amplitudes and periods that converge to some non zero limits.
(ii) A Poincar\'e-Bendixson like theorem holds, and as a consequence
the global attractor of (\ref{eqmain}) is composed of equilibria,
periodic solutions, and their unstable manifolds.  (iii) The global
attractor of (\ref{eqmain}) has a Morse decomposition, it is ordered
by a discrete Lyapunov functional, and it is composed only of
equilibria of (\ref{eqmain}), the periodic solutions described in the
first conjecture, and connections from more-rapidly oscillating
periodic solutions to less-rapidly oscillating periodic solutions.
\par 

\medskip
\par With few hypotheses on the feedback $f$, we have shown that a
local Hopf-bifurcation theorem of \cite{eichmann} applies to
(\ref{eqmain}).  This gives an essential element in the proof of
conjecture (i): for some decreasing sequence $\epsilon_k \rightarrow
0$, a Hopf bifurcation occurs at the zero equilibrium of
(\ref{eqmain}) each time that $\epsilon$ crosses one $\epsilon_k$,
which gives rise to an oscillating periodic solution with $2k$ zeroes
per period.  The corresponding theorem is rigorously stated in section
\ref{sub_hopf} and proved in appendix \ref{subsec_hopf}.
Admitting the existence of branches corresponding to these periodic
solutions, in section \ref{subsub_a}, we show that if we have
$\eta(\epsilon) = c \epsilon + o(\epsilon)$ in (\ref{eq_delay_eta})
for some $0< c < + \infty$, then the amplitudes of all periodic
solutions converge to the same non zero limit when $\epsilon
\rightarrow 0$.  On the contrary, following a proposition of
\cite{MPN2}, we show in section \ref{subsub_b} that when $\eta(0)\neq
0$, the amplitude of periodic solutions with $2k$ zeroes per period
(along the branch that appears at $\epsilon_k$), is bounded from above
by some constant $C_k$ (independent of $\epsilon$) such that $C_k
\underset{k \rightarrow + \infty}{\longrightarrow } 0$.

\subsection{Sequence of Hopf-Bifurcations for Eq. (\ref{eqmain})}
\label{sub_hopf}

The PhD thesis of M. Eichmann \cite{eichmann} contains a local
Hopf-bifurcation theorem for state-dependent DDEs which implies the
following. 

\begin{theorem} \label{thm_hopf} Suppose that $f$ is $C^2(\mathbb{R},
  \mathbb{R})$, $r : ]0,1[ \times C^0([-M,0], \mathbb{R}) \rightarrow
  \mathbb{R}$ is $C^1$, and $r : ]0,1[ \times C^1([-M,0], \mathbb{R})
  \rightarrow \mathbb{R}$ is $C^2$, for $M>0$.  Suppose that $f(0) =
  0$ and $\vert f^\prime (0) \vert >1$.  Then there is a decreasing
  sequence $(\epsilon_k)_{k \in \mathbb{N}}$ converging to zero, such
  that for any $k \geq 0$ there is an open interval $]- \eta_k ,
  \eta_k[$ and $C^1$ mappings $y^* : ]- \eta_k , \eta_k[ \rightarrow
  C^{1}([-M,0], \mathbb{R}) $, $\epsilon^* : ]- \eta_k , \eta_k[
  \rightarrow ]0,1[ $ and $w^* : ]- \eta_k , \eta_k[ \rightarrow
  \mathbb{R}$, with $y^*(0)=x^*=0$, $\epsilon^*(0)= \epsilon_k$ and
  $w^*(0) = \beta_k = {\rm {Im}}(\lambda_k)$ such that for any $u \in ]-
  \eta_k, \eta_k[$ there is a periodic solution to
$$ \epsilon^*(u) x^\prime(t) = - x(t) + f(x(t- r(\epsilon^*(u), x_t))) $$
with initial condition $x_0 = y^*(u)$ and with frequency
$\frac{w^*(u)}{2 \pi}$.
\end{theorem}
We remark that the delay function $r(\ep, x_t)$ in the theorem above
is more general than that in Eq.  (\ref{eq_delay_eta}), and that the
function $f$ does not have to be positive or negative feedback.

The theorem is proved in appendix \ref{subsec_hopf}. Here we first
justify the existence of the critical values $\ep_k$, and show that
the associated frequencies $\frac {w_k}{2\pi} = \frac {\beta_k}{2\pi}
\underset{k \to \infty}{\longrightarrow} + \infty $.
Suppose that $f(0)=0$ and that $|f^\prime(0)|>1$, and let  
$x^*=0$  be the unstable steady solution of Eq. (\ref{eqmain}).
The linearization of  Eq. (\ref{eqmain}) at  $x^*=0$ is
$$ \epsilon y'(t) = - y(t) + f^\prime(0) y(t- r_0)  $$
 where $r_0 = r(0)= r(x^*)=1$.
The characteristic equation associated to this linearized equation 
is $1 + \epsilon \lambda = f^\prime(0) e^{- \lambda}$, or, equivalently, with 
$\lambda = \alpha + i \beta \in \mathbb{C}$
$$ \left \{
\begin{array}{ccc}
1 + \epsilon \alpha &=& f^\prime(0) e^{-\alpha} \cos(\beta) \\
   \epsilon \beta &=& - f^\prime(0) e^{- \alpha} \sin{\beta},
 \end{array} \right . $$
This is the same characteristic equation as for the constant-delay equation, and there exists  a sequence 
$\epsilon_k \underset{k \rightarrow + \infty}{\longrightarrow} 0$ such
that for each $\ep=\ep_k$ 
the characteristic equation has a single pair of solutions on the 
imaginary axis
$\lambda =\pm i \beta_k$, $\bt_k>0$. Moreover,  
$\bt_k\to\infty$ as $k\to\infty$. 
Since $\bt_k$ is the angular frequency of the periodic orbit unfolded
at $\ep_k$, theorem \ref{thm_hopf} implies the existence of rapidly
oscillating periodic solutions of Eq. (\ref{eqmain}) as $\epsilon$
tends to zero.

\subsection{Case $\eta(0)=0$ with $0< \eta^\prime(0) < + \infty$ }
\label{subsub_a}

With the previous assumptions on $\eta$, one has
$r(x,\epsilon)= 1 + \eta(\epsilon) R(x) \sim 1 + \epsilon
\eta^\prime(0) R(x) $ when $\epsilon \rightarrow 0$. Without loss of
generality, we assume $\eta^\prime(0) = 1$, so that the delay is $r(x,
\epsilon) = 1 + \epsilon R(x)$ in Eq. (\ref{eqmain}).

Suppose that there is an $\epsilon_0$ for which Eq. (\ref{eqmain}) has
a periodic solution $x_0(t)$ with period $T_0$.  Given an integer
$n>0$, Eq. (\ref{eqmain}) and the periodicity of $x_0(t)$ imply that
\[
\epsilon_0 \dot x_0(t)=-x_0(t)+f(x_0\{t-1-nT_0-\epsilon_0 R[x_0(t)]\})
\; .
\]
Then, rescaling time as $\hat t= t/(1+nT_0)$, we obtain that $x_n(
t)=x_0[(1+nT_0)\hat t]$ satisfies the equation
\[
{\ep_0\over 1+nT_0} {d x_n\over d \hat t}(\hat t)=
-x_n(\hat t)+f\left(x_n\left\{\hat t-1-{\ep_0\over  1+nT_0}  R[x_n(\hat t)]\right\}
\right).
\]
Therefore, for $\ep=\ep_n= \ep_0/(1+nT_0)$ Eq. (\ref{eqmain}) admits
the periodic solution $x_n(t)=x_0[t(1+nT_0)]$ with period
$T_n=T_0/(1+nT_0)$. 
It shows that the branches of rapidly
oscillating periodic solutions can be obtained from the first branch
of periodic solutions. 
Hence, assuming that the first branch exist up to $\epsilon \rightarrow 0$, 
it follows that all the other branches exist,
and the amplitude of periodic solutions along all theses branches 
converge to the same positive limit when $\epsilon \rightarrow 0$. 

In this sense, Eq. (\ref{eqmain}) sustains ``large amplitude'' rapidly
oscillating periodic solutions as $\ep$ tends to zero when the delay
function is of the form $r(\epsilon, x) = 1 + \epsilon R(x)$ or
$r(\epsilon, x) = 1 + \epsilon \eta^\prime(0) R(x)$.  As mentioned in
section \ref{metastability}, this is one of the signatures of the
existence of metastable solutions in the case of DDEs with constant
delay.  So this strengthens the similarity of DDEs with state
dependent delay and DDEs with constant delay, thus giving support to
the possibility of metastability in the case $\eta(0)=0$ with $0<
\eta^\prime(0) < + \infty$.

\subsection{Case $\eta(0)\neq 0$}
\label{subsub_b}

The situation for $\eta(0) \neq 0$ is different from the one depicted
above, and this can be understood thanks to proposition 3.4 in
\cite{MPN2}, which precludes the existence of large amplitude rapidly
oscillating periodic solutions as $\ep\to 0$.  First, we recall that
proposition in \cite{MPN2}, and then we discuss its consequences in
terms of the amplitude of the periodic solutions.

\begin{proposition}[Mallet-Paret, Nussbaum]
\label{propMPN}
Suppose that the feedback $f$ is $C^0(\mathbb{R},\mathbb{R})$ 
and that the delay
function $r(\epsilon, x)$ is Lipschitz regular in $x$.
Let $x(t)$ satisfy equation {\rm (\ref{eqmain})} for $t\in\R$ 
for some value of $\epsilon$, and suppose
there exist $a_0<a_1$ and $t_0<t_1<t_2<t_3$ such that
$x(t_i)\le a_0$ for  even $i$, and $x(t_i)\ge a_1$ for odd $i$, for
$0\le i\le 3$. Then
\[
\max_{[a_0,a_1]}r(\cdot,\ep)-\min_{[a_0,a_1]}r(\cdot,\ep)\le 3(t_3-t_0).
\]
\end{proposition}

\begin{proposition}
\par Suppose that the feedback $f$ is $C^0(\mathbb{R},\mathbb{R})$ 
and that the delay
function $r(\epsilon, x)$ is Lipschitz regular in $x$.
Let $r(x,\ep)=1+ \eta (\ep ) R(x)$, and suppose that
$R(x)$ is not constant on any interval (e.g. $R(x)=r_kx^k$, $k\ge 1$
and $r_k\ne 0$) and $\eta(\ep) \sim \eta(0) \neq 0$.
\par Then, there is a function $\varphi$, depending only on $\eta(0)$ and $R$, 
with $\varphi(T) \underset{T \rightarrow 0}{\longrightarrow} 0$,
such that for any $x(t)$ be a periodic solution of
Eq. (\ref{eqmain}) with period $T$, and
 $a_0=\min x(t)$ and $a_1 = \max x(t)$, we have
 $$ \vert a_1 - a_0 \vert \leq \varphi(T).  $$
\par In particular, if $f$ and $r$ satisfy additionally the hypotheses of theorem \ref{thm_hopf}, 
the periodic solutions $x_k$ that appear when $\epsilon = \epsilon_k$ (see theorem \ref{thm_hopf}) 
have periods $T_k \underset{k \to \infty}{\longrightarrow}  0$ 
and amplitudes $ \max \{ x_k(t) \} - \min \{ x_k(t) \}  \underset{k \to \infty}{\longrightarrow} 0$.
\end{proposition}
\begin{proof}
To see why proposition \ref{propMPN} precludes the existence of large-amplitude rapidly
oscillating periodic solutions, let $x(t)$ be a periodic solution of
Eq. (\ref{eqmain}) with period $T$, and
let $a_0=\min x(t)$ and $a_1 = \max x(t)$. Choose $t_0$ such that
$x(t_0)=a_0$, choose $t_1$ such that $t_0<t_1 < t_0 +T$ and
$x(t_1)=a_1$, and $t_2=t_0 +T$ and $t_3=t_1+T$.  As in section
\ref{intro}, let $r(x,\ep)=1+ \eta (\ep ) R(x)$, and suppose that
$R(x)$ is not constant on any interval, e.g. $R(x)=r_kx^k$, $k\ge 1$,
$r_k\ne 0$, and $\eta(\ep) \sim \eta(0) \neq 0$.  Then $6T \geq
3(t_3-t_0)\ge \eta(0)|r_k| \Dt_R/3$, where
$\Dt_R=\max_{x\in[a_0,a_1]}R(x)-\min_{x\in[a_0,a_1]}R(x)>0$.  Given
that $\eta(0) \neq 0$, we have $|r_k| \Dt_R/3 \leq
\frac{6T}{\eta(0)}$. 

This estimate relates the period of oscillations indirectly to their
amplitude (through $\Dt_R$) :
the faster the periodic oscillations are, i.e. the smaller $T$ is, the
smaller their amplitude is.
The Hopf bifurcation at $\ep = \epsilon_k$ 
gives rise to a periodic solution with period $T_k = \frac 1{\beta_k} \underset{k \to \infty}{\longrightarrow} 0$ 
(see theorem \ref{thm_hopf} and  appendix \ref{subsec_hopf}), 
and the estimate above implies that so does their amplitudes $a_1 - a_0 \underset{k \to \infty}{\longrightarrow} 0$. 
\end{proof}

We argue that this result indicates that DDE (\ref{eqmain}) with state
dependent delay, and $\eta(0)\neq 0$, cannot support metastable
transient oscillations that resemble those of DDEs with constant delays. 
This is further supported by the fact that the
profiles of oscillations in this case are not solutions to usual transition
layer equations (see discussion in next section \ref{sec_tle}) and
confirmed through extensive numerical investigations (section
\ref{section_num}).

\section{Transition layer and metastability}
\label{sec_tle}

In this section, we refine our previous analysis of the conditions
under which DDE (\ref{eqmain}), with state dependent delay, can support
metastable oscillations through the introduction of transition layer
equations. Such equations have been used previously to determine the
shape of slowly oscillating periodic solutions for scalar DDEs with
constant delay and negative feedback in the singular limit $\ep \to
0$ \cite{m-pn}.  They have also been instrumental for the analysis of
metastable solutions in scalar DDEs with constant delays and monotone
feedback in the same singular limit \cite{jdde}. 

For negative feedbacks, the singular limit as $\ep\to 0$ of Eq.
(\ref{eqmain}) in the case $ \eta(0) \neq 0$ can be analyzed through
the theory developed in \cite{MPN1,MPN2,MPN3} that replaces transition
layer equations with the so-called ``Max-Plus'' equations. However,
given that, as argued in section \ref{subsub_b} 
and numerically shown in section \ref{section_num}, such systems do not
support metastability, we will not dwell any further in this case.
Throughout the remainder of this section, our focus is on the case
$\eta(0)=0$, which we henceforth assume to hold.

For some feedback functions $f$ and state dependent delay function
$r$, solutions of Eq. (\ref{eqmain}) have an approximately periodic
square-wave shape when $\epsilon \rightarrow 0$, as in the constant
delay case $r(x,\epsilon) = 1$, and the corresponding ``jumps'' can be
analyzed with the help of transition layer equations.
The approximate period of metastable oscillations (depending on
$\epsilon$) is an essential point in finding transition layer
equations.  For constant delay, at first order, this period is
$2+\rho\ep$ for negative feedback (\cite{m-pn} theorem 3.2), and
$1+\epsilon \rho$ for positive feedback \cite{jdde}.  In section
\ref{section_num} similar asymptotics are shown to hold in the
state-dependent delay case as well when $\eta^\prime(0)= 0 $ and $ 0
\leq \eta^\prime(0) < + \infty$ .  In this section we write
appropriate transition layer equations under various hypotheses on
$\eta$, show that their solutions exist, and that these equations can
be used to characterize metastability, as confirmed by the numerical
investigation presented in section~\ref{section_num}.

The case $\eta(0)=0$ and $\eta^\prime(0) = + \infty$ has also been
investigated, using $r(x, \epsilon) = 1 + \epsilon^\alpha R(x)$ with
$\alpha=\frac{1}{2}$. We found numerically that oscillations are
square-wave-like when $\epsilon$ is small, their period is either $2+
\epsilon^{\alpha} \rho$ (negative feedback) or $1+ \epsilon^{\alpha}
\rho$ (positive feedback), and rescaling time as $\phi(t) =
x(\frac{t}{\epsilon^{\alpha}})$, one observes convergence to a
transition layer profile (see figure~\ref{Fig_front_a05} in
section~\ref{section_num}).  However the scaling argument used in the
case $\eta^\prime(0) < + \infty$ does not apply here, and these
transition layer profiles are not analyzed in this section.

The remainder of this section is organized as follows.  First, in
section \ref{subsection_tle_alpha=1}, we show that when $ 0 <
\eta^\prime(0) < + \infty$, appropriately defined transition layer
equations can be used to find a symmetry condition that characterizes
precisely the cases of metastable oscillatory transients.  Details
about the existence of transition layer solutions, their numerical
construction, and illustrating figures can be found in the appendix
section \ref{app_tlo}.  The case of $\eta^\prime(0)=0$ is examined in
section \ref{subsection_tle_alpha>1}. One can still write a transition
layer problem, but it does not depend on the delay function $R$
anymore, leading to incorrect results. To overcome this, we have
introduced a one-parameter family of auxiliary transition layer
problems, and thanks to the analysis of the corresponding
one-parameter family of transition layer solutions we are able to
characterize the cases where metastability can occur. Finally, in
section \ref{subsec_meta}, we discuss a new phenomenon: the
possibility of a state dependent delay giving rise to metastability in
equations that do not exhibit such transients when the delay is constant.

\subsection{Case $\eta(0)=0$ and $0< \eta^\prime(0) < + \infty$}
\label{subsection_tle_alpha=1}

\subsubsection{Positive feedback}
\label{subsub_pf}

Let $f$ be of positive feedback type, and consider Eq.
(\ref{eqmain}) where $r(x,\epsilon) = 1 + \eta(\epsilon)R(x)$ with
$\eta(0)=0$ and $0< \eta^\prime(0) < + \infty$ (without loss of
generality we assume that $\eta^\prime(0)=1$ in the remainder of this
section). For such delays, our numerical investigations show that
metastable oscillations are approximately $1+ \epsilon \rho$ periodic.

As for DDEs with constant delays (see section \ref{metastability}),
the jumps of these square-wave-like metastable solutions
connecting respectively $b$ to $-a$ and vice versa
are described by:
\begin{equation}\label{eqTLE}
  \dot  \phi^\pm(t)  = - \phi^\pm(t) + f( \phi^\pm( t - R(\phi^\pm(t))
  + \rho^\pm  ) )
  \; \; \; \; 
\end{equation}
where $\rho^->0$ and $\rho^+>0$ are unknown constants (drift speeds) and the
functions $\phi^\pm$ satisfy the boundary conditions $\lim_{t\to
  -\infty}\phi^-(t)=b$, $\lim_{t\to \infty}\phi^-(t)=-a$ , $\lim_{t\to
  -\infty}\phi^+(t)=-a$, and $\lim_{t\to \infty}\phi^+(t)=b$. 

If a solution to equation (\ref{eqTLE}) exists, it is called a
transition layer solution.  In contrast to the constant delay-case,
this transition layer equation (\ref{eqTLE}) is a state-dependent
equation, and from a theoretical point of view, depending on the
values of $\rho$ and $R(x)$, it may be both advanced and delayed. The
numerical method used for solving Eq. (\ref{eqTLE}) is presented in
the Appendix \ref{sub_nummeth}. We have defined an operator $\mathcal
T$ whose (stable locally attractive) fixed points are
solutions of (\ref{eqTLE}).

In general, the constants $\rho^-$ and $\rho^+$ associated to the
decreasing and increasing transition layer solutions are
different. In the Table \ref{Fig_tableau_rho} we display 
$\rho^-$ and $\rho^+$ solutions of equation (\ref{eqTLE}) for
$\eta(\epsilon)=\epsilon$ and various choices of $R(x)$, for both
symmetric and non-symmetric positive feedback $f$ (numerical method
details given in the Appendix \ref{sub_nummeth}). 
We found that in the state dependent case of equation (\ref{eqTLE}), 
as in the constant delay case, oscillatory transients are metastable only when
$\rho^+ = \rho^-$. 
Table \ref{Fig_tableau_rho} also shows that
 $\rho^+ =\rho^-$ is obtained only when the positive feedback function is
symmetric and the delay $R(x)$ is even. 

\begin{table}[th!]
\begin{center}
\begin{tabular}{|c|c|c|c|c|}
  \hline
  (a)      &  $R(x)=0$   & $R(x)=x$       & $R(x)=\cos(x)$ & $R(x)=\frac12 x(1+x)$  \\ \hline
  $\rho^+$     &   0.824     &     0.554      &       1.752    &          0.732         \\ \hline
  $\rho^-$     &   0.824     &     1.158      &      1.752      &         1.158         \\ \hline
\end{tabular}
\end{center}
\vspace{-0.6cm}
\begin{center}
\begin{tabular}{|c|c|c|c|c|}
\hline
    (b)      &  $R(x)=0$   &    $R(x)=x$    & $R(x)=\cos(x)$ & $R(x)=\frac12 x(1+x)$  \\ \hline
$\rho^+$     &      1.024  &      0.690     &       1.916    &        0.932           \\ \hline
$\rho^-$     &     0.664   &      0.994     &       1.612    &        0.884           \\ \hline
\end{tabular}
\end{center}
\vspace{-0.3cm}
\caption{\small{Drift speeds $\rho^+, \rho^-$,
    solutions of equation (\ref{eqTLE}) for positive feedbacks $f$, $\eta(\epsilon)=\epsilon$ and
    various choices of delay $R(x)$. Metastability occurs only when 
    $\rho^+ = \rho^-$.
    { Table (a) Symmetric Positive Feedback 
      $f(x) = \frac12 \arctan (5x)$. Table (b) Non-symmetric Positive Feedback 
      $f(x) = \frac12 \arctan (5(x-0.05)) + \frac12 \arctan(0.25)$.
      (see appendix \ref{sub_nummeth} for details on the numerical
      method and parameters value used.) 
    }}}
\label{Fig_tableau_rho}
\end{table}

\subsubsection{Negative feedback}
\label{subsub_nf}

We now discuss the transition layer equation for state dependent
delayed negative feedback.  Numerically, transient oscillations are
square-wave-like and they have an approximate period $T \approx 2 +
\epsilon C$ (see section \ref{section_num}). Likewise the case of DDEs
with constant delay, the transition layer equations for the increasing
and decreasing transition layer solutions of DDEs with state dependent
delays are coupled :
\begin{equation}
\left \{
\begin{array}{ccc}
\dot \phi^+(t) &=& - \phi^+(t) + f( \phi^-[ t - R(\phi^-(t)) + \rho^-
] ) \; , \\
\dot \phi^-(t) &=& - \phi^-(t) + f( \phi^+[ t - R(\phi^+(t)) + \rho^+
] ) \; , \\
\end{array}
\right .
\label{eqTLEneg}
\end{equation}
where $\phi^+$ is increasing and $\phi^-$ is decreasing on
$\mathbb{R}$, with $\lim_{t\to -\infty}\phi^+(t)=-a$, $\lim_{t\to
  \infty}\phi^+(t)=b$, $\lim_{t\to -\infty}\phi^-(t)=b$ and
$\lim_{t\to \infty}\phi^-(t)=-a$, with $\phi^+(0)=\phi^-(0)=0$ and
$\rho^\pm$ are unknown real constants (drift speeds).
See appendix section \ref{app_tlo} for numerical
solutions of Eq. (\ref{eqTLEneg}).

For negative feedbacks, the symmetry condition supporting
metastability is always satisfied when the delay is constant
\cite{jdde}, and we show that the same holds when the delay is state dependent.
In Table \ref{Fig_tableau_rho_NF} we display the drift speeds $\rho^+$,
$\rho^-$ that are solutions of Eq. (\ref{eqTLEneg}) for
$\eta(\epsilon)=\epsilon$ and various choices of $R(x)$. Comparison of
Tables~\ref{Fig_tableau_rho} and \ref{Fig_tableau_rho_NF} shows that
when the feedback function $f$ is symmetric, and $R(x)$ is even,
$\rho^+=\rho^-$ for both positive and negative feedbacks.  This
happens because when $f$ is symmetric, and $R(x)$ is even, the
increasing solutions of the transition layer equations for both
positive and negative feedback coincide (the same happens for the
decreasing solutions).
 
\begin{table}[th!]
\begin{center}
\begin{tabular}{|c|c|c|c|c|}
\hline
    (a)      &  $R(x)=0$   & $R(x)=x$       & $R(x)=\cos(x)$ & $R(x)=\frac12 x(1+x)$  \\ \hline
$\rho^+$     &    0.824    &    1.172       &    1.744       &     1.122             \\ \hline
$\rho^-$     &    0.824    &    0.434       &    1.744       &     0.714             \\ \hline
\end{tabular}
\end{center}
\vspace{-0.6cm}
\begin{center}
\begin{tabular}{|c|c|c|c|c|}
\hline
       (b)      &  $R(x)=0$   & $R(x)=x$       & $R(x)=\cos(x)$ & $R(x)=\frac12 x(1+x)$  \\ \hline
$\rho^+$     &   -0.702     &    -0.616      &    0.280        &      -0.646            \\ \hline
$\rho^-$     &    2.574     &    2.362       &     3.570      &      2.476           \\ \hline
\end{tabular}
\end{center}
\vspace{-0.3cm}
\caption{\small{Drift speeds $\rho^+, \rho^-$,
    solutions of equation (\ref{eqTLEneg}) for negative feedbacks $f$, $\eta(\epsilon)=\epsilon$ and
    various choices of delay $R(x)$. Metastability occurs regardless of the equality  $\rho^+
    =\rho^-$.  {(a) Symmetric
      negative feedback $f(x)= - \frac12 \arctan(5x)$, (b)
      Non-symmetric negative feedback $f(x) =
      - \frac12 \arctan(5(x+0.05)) + \frac12 \arctan(0.25)$.
      (see Appendix \ref{sub_nummeth_neg} for details on the
      numerical method and parameters value used.)   }}}
\label{Fig_tableau_rho_NF}
\end{table}

\subsection{Case $\eta(0)=0$ and $\eta^\prime(0)=0$}
\label{subsection_tle_alpha>1}
\subsubsection{Positive Feedback}
\label{subsubsection_tle_alpha>1-pos}

If $\eta(0)=0$ and $\eta^\prime(0)=0$, solutions are approximately
$1+\epsilon \rho$ periodic, 
and the same time rescaling $\phi(t) = x(\epsilon t)$
implies
$ \dot  \phi(t)  = - \phi(t) + f( \phi(t -
\frac{\eta(\epsilon)}{\epsilon}  R(\phi(t)) + \rho  ) ) \; , $
and as $\epsilon \rightarrow 0$
\begin{equation}\label{eqTLEcst}
\dot  \phi(t)  = - \phi(t) + f( \phi( t + \rho ) ) \; ,
\end{equation}
where the decreasing and increasing solutions $\phi^\mp$ must satisfy
the boundary conditions $\underset{t \rightarrow - \infty}{\lim}
\phi^+(t) = -a$, $\underset{t \rightarrow + \infty}{\lim} \phi^+(t) =
b$, $\underset{t \rightarrow - \infty}{\lim} \phi^-(t) = b$ and
$\underset{t \rightarrow + \infty}{\lim} \phi^-(t) = -a$, with
$\phi^\pm(0)=0$ and $\rho = \rho^\mp>0$ is an unknown constant.
This equation (\ref{eqTLEcst}) is the same one found in the 
constant delay case, for which the existence of decreasing and
increasing transition layer solutions has been proven in \cite{jdde}.
In particular, when $\eta(0)=0$ and $\eta^\prime(0)=0$, the
drift speeds $\rho^\mp$ are equal to those of the corresponding
constant delay case ($R(x)=0$). As a consequence, we obtain that for
positive feedback $f$, if $\eta(0)=0$ and $\eta^\prime(0)=0$,
metastability cannot occur if $\rho^+ \neq \rho^-$. However, when
$\rho^+=\rho^-$, metastability may or may not occur. 

To obtain the symmetry requirement for metastability in this case
$\eta(0)=0$ and $\eta^\prime(0)=0$, we introduce the following
1-parameter family of transition layer equations
\begin{equation}\label{eqTLElambda}
\dot  \phi^\pm(t)  = - \phi^\pm(t) + f( \phi^\pm[ t + \rho_\lambda^\pm
- \lambda R(\phi^\pm(t)) ] ), 
\end{equation}
where, $\lambda \in \mathbb{R}$ is a real parameter. As previously,
$\rho_\lambda^\pm$ are some unknown real constants, and $\phi^\pm$ are
the transition layer solutions, that are expected to depend on
$\lambda$.  The case $\lambda=0$ reproduces the drift speeds
$\rho^\pm$ of the transition layer equation (\ref{eqTLEcst}).  In
figure \ref{Fig_rho_pm_alpha>1} we display the constants
$\rho_\lambda^\pm$ as function of $\lambda$ for the symmetric positive
feedback function $f(x)= \frac12 \arctan(5x)$ and various choices of
$R(x)$. We found that if $\rho_\lambda^+ = \rho_\lambda^-$ holds only
for $\lambda = 0$ (panels (b) and (d) in figure
\ref{Fig_rho_pm_alpha>1} ), metastable DITOs are not observed for
positive values of $\epsilon$. Oscillatory transients are metastable
only when $\rho^+_\lambda = \rho^-_\lambda$ on some non-trivial
interval $\lambda \in [0, \delta]$ with $\delta >0$ (panels (a) and
(c) in figure \ref{Fig_rho_pm_alpha>1} ). So this is the new
sufficient condition for the existence of metastable oscillatory
transients when the positive feedback $f$ is symmetric and
$\eta^\prime(0)=\eta(0)=0$.

\begin{figure}[th!]
\begin{center}
\includegraphics[width=7.5cm]{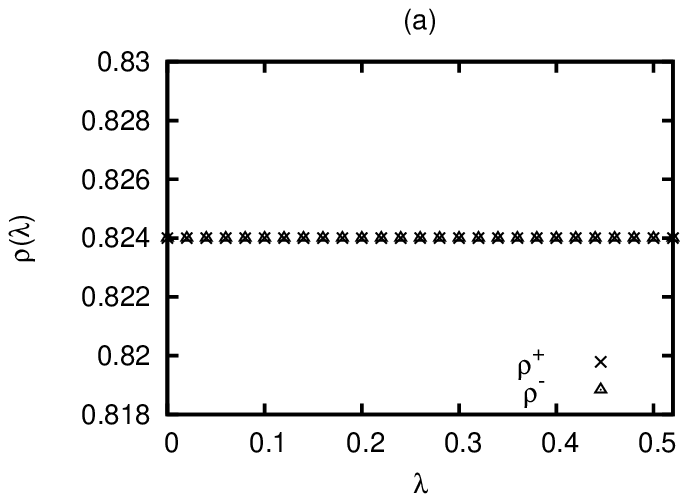}
\includegraphics[width=7.5cm]{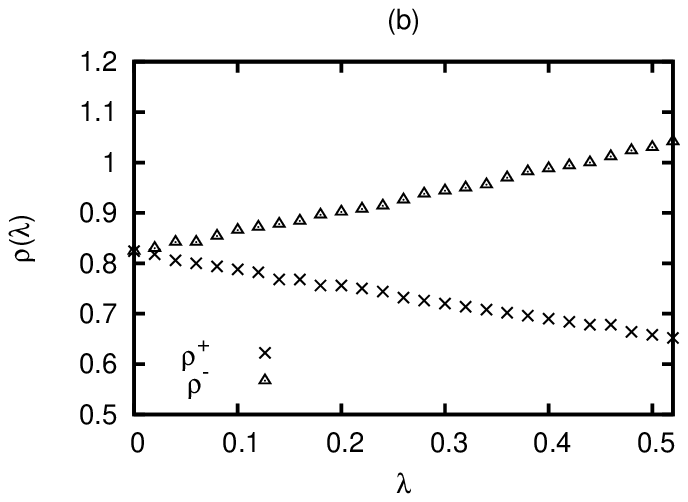}  
\end{center}
\vspace{-1.0cm}
\begin{center}
\includegraphics[width=7.5cm]{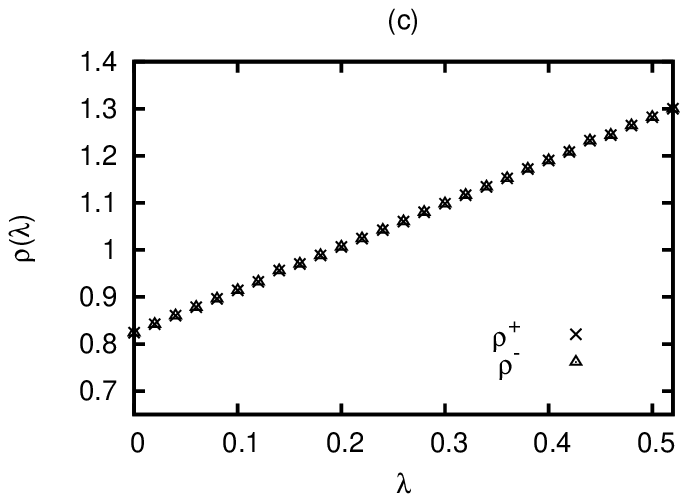}
\includegraphics[width=7.5cm]{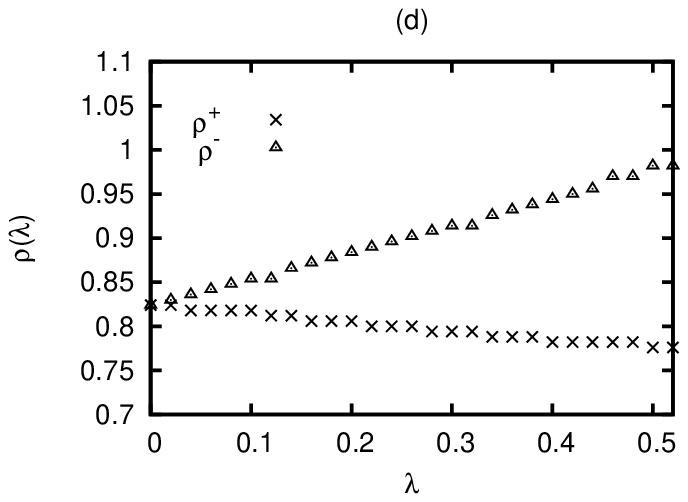}  
\end{center}
\vspace{-1.0cm}
\caption{\small{Drift speeds $\rho_\lambda^\pm$, solutions of equation
  (\ref{eqTLElambda}), for $\lambda \in [0,0.5]$, using the symmetric
  positive feedback function $f(x)= \frac12 \arctan(5x)$ and various
  choices of $R(x)$: (a) $R(x)=0$ (constant delay); (b) $R(x)=x$; (c)
  $R(x) = \cos(x)$; (d) $R(x) = \frac12 x(1+x)$.  
See appendix \ref{sub_nummeth} for details on numerical methods and
  parameters used.
}}
\label{Fig_rho_pm_alpha>1}
\end{figure}

\subsubsection{Negative Feedback}
\label{subsubsection_tle_alpha>1-neg}
For $\eta(0)=0$ and $\eta^\prime(0) \in ]0;+ \infty[$, we have seen
that the symmetry condition supporting metastability always holds when
the feedback is negative.  The same is true if
$\eta^\prime(0)=\eta(0)=0$: regardless of the symmetry of the feedback
$f$ and of the function $R(x)$ in $r(x,\epsilon) = 1 + \eta(\epsilon)
R(x)$, rapidly oscillating transients are metastable (see section
\ref{section_num}).

\subsection{Metastability induced by state dependent delay}
\label{subsec_meta}

\par In this section we present a new phenomenon:  given a constant DDE
that does not display metastability, it is possible to add a state
dependence to the delay so that the resulting state dependent DDE will
exhibit metastability.
To this end, we consider equation (\ref{eqmain}) with a non-symmetric
positive feedback function $f$, and delay function $r_\lambda(x,
\epsilon) = 1 + \epsilon \lambda R(x)$, so that  $\lambda=0$ 
corresponds to constant delay DDE which does not exhibit metastability.

We did a numerical investigation using the non-symmetric positive
feedback function $f(x) = \frac12 \arctan(5(x+0.05)) - \frac12
\arctan(0.25)$.  We have used the following functions $R(x)$: $R(x) =
x$, $R(x) = \cos(x)$, $R(x) = \frac12 x(1+x)$. Solving the transition
layer equation (\ref{eqTLElambda}) for each function $R$ we have
numerically computed the $\lambda$~-families of constants
$\rho^+_\lambda$ and $\rho^-_\lambda$, the parameter $\lambda$ being
varied within the interval $[- 1.0, 1.0]$ (it should be remarked that
for large $\lambda$ values the numerical solution of the transition
layer equation (\ref{eqTLElambda}) is problematic). Results are
displayed in the figure~\ref{Fig_rho_3}. Metastability will occur for
those values of $\lambda$ such that
$\rho^+_\lambda=\rho^-_\lambda$. Figure~\ref{Fig_rho_3}(c)
($R(x)=\cos (x)$) shows that no solution was found such that
$\rho^+_\lambda=\rho^-_\lambda$, likewise the constant delay case
(figure~\ref{Fig_rho_3}(a)), indicating that introducing a state
dependent delay may not make up for the lack of symmetry of the
feedback function $f$. Nevertheless, figure~\ref{Fig_rho_3}(b)
($R(x)=x$) and figure~\ref{Fig_rho_3}(d) ($R(x) = \frac12 x(1+x)$)
show that, in these cases, adding state dependence to the delay has
resulted in metastability. The solution such that
$\rho^+_{\lambda_c}=\rho^-_{\lambda_c}$ is $\lambda_c \approx 0.5$ in
the case $R(x) = x$ (see figure~\ref{Fig_rho_3}(b)), and $\lambda_c
\approx 1.1$ in the case $R(x) = \frac12 x (1+x)$ (see
figure~\ref{Fig_rho_3}(d)).

\begin{figure}[th!]
\begin{center}
\includegraphics[width=7.5cm]{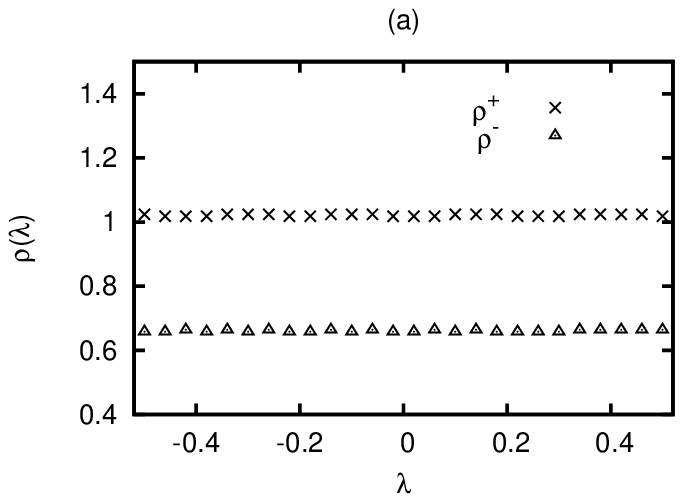}
\includegraphics[width=7.5cm]{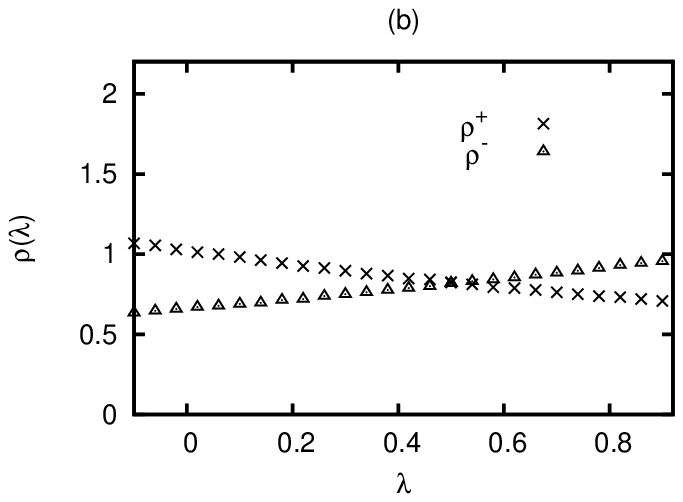}  
\end{center}
\vspace{-1.0cm}
\begin{center}
\includegraphics[width=7.5cm]{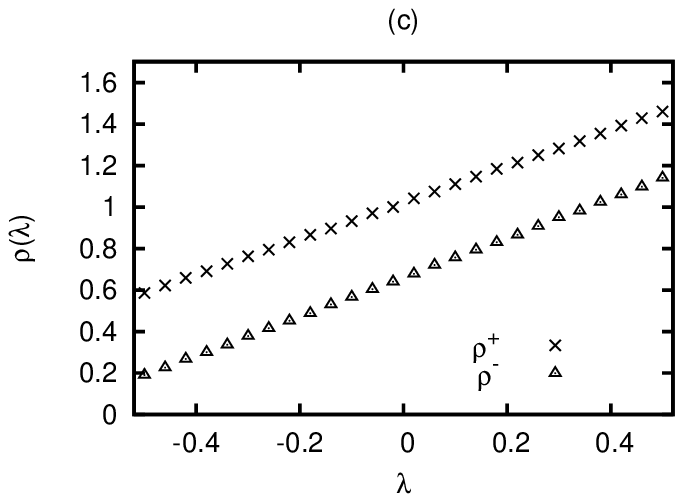}
\includegraphics[width=7.5cm]{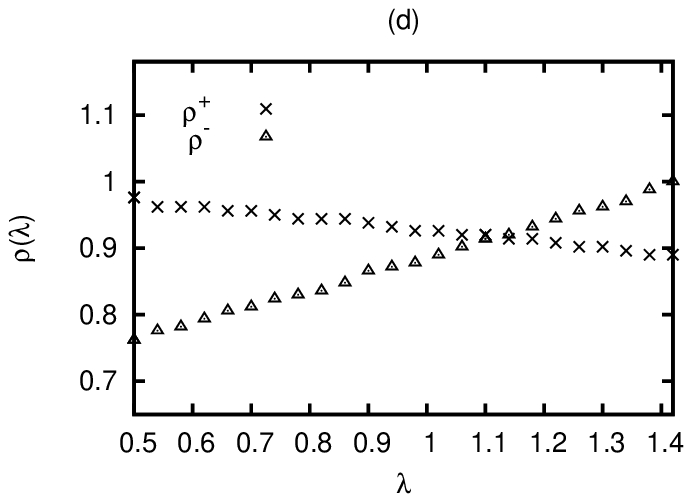}  
\end{center}
\vspace{-1.0cm}
\caption{\small{Drift speeds $\rho^\pm$ as function of $\lambda$ for
  non-symmetric positive feedback function $f(x)= \frac12
  \arctan(5(x+0.05)) - \frac12 \arctan(0.25)$, and  delay
  function $R(x)$: (a) constant delay case $R(x)=0$; (b)
  $R(x) = x$; (c) $R(x)= \cos(x)$; and  (d) $R(x) = \frac12 x
  (1+x)$.  (See appendix section \ref{sub_nummeth} for details on
  numerical methods and parameters value used.)
  }}
\label{Fig_rho_3}
\end{figure}

\section{Numerical simulations of equation (\ref{eqmain})}
\label{section_num}

To corroborate the characterization of metastable state-dependent
DITOs obtained in sections \ref{sec:periodic_sol} and \ref{sec_tle},
we have carried a numerical investigation of equation (\ref{eqmain}),
thus completing the analysis of the transient dynamics.

We shall present the results of numerical solutions of Eq.
(\ref{eqmain}), for negative and positive feedback functions $f$, and
delay functions of the form $r(x,\epsilon)= 1+ \eta(\epsilon) R(x)$.
We have used $\eta(\epsilon)=\epsilon^\alpha$ with $\alpha= 0$
($\eta(0) \neq 0$), $\alpha= 1$ ($\eta(0)= 0,\eta^\prime(0)>0 $) ,
$\alpha = \frac12$ ($\eta(0)= 0,\eta^\prime(0)=\infty $) and $\alpha
>1 $ ($\eta(0)= 0,\eta^\prime(0)=0 $). As for $R(x)$, we have
considered the following cases: $R(x) = x$, $R(x) = x^2$, $R(x)
=\frac{1}{2} x(1+x)$, $R(x) = \sin(x)$, and $R(x) = \cos(x)$.

For positive feedback, with constant delay, metastability requires
that the feedback function is an odd function of $x$, so in this case
we have used $f(x) = \frac{1}{2} \arctan(5x)$.  For negative feedback
case we have used both symmetric $f(x) = - \frac{1}{2} \arctan(5x)$
and non-symmetric $f(x) = - \frac{1}{2} \arctan(5(x+0.05)) + \frac12
\arctan(0.25)$.  The results are qualitatively the same for these two
functions, so we shall only show the results for the symmetric
negative feedback function.

In the case $\eta(0)= 0$ and $\eta^\prime(0)=0$ ($\alpha > 1$), the
observations are qualitatively the same. As $\alpha$ increases, the
results are closer and closer to those observed in the constant delay
case.

The numerical results were checked using first and second order
numerical schemes, using time steps $dt=5.10^{-5} $ and $dt=2.
10^{-6}$, and with linear interpolation for the state dependent delay
function. Simulations of solutions of equation (\ref{eqmain}) have
also been checked using the RADAR-V package in Fortran. The range of
$\epsilon$ values we investigated is $\epsilon \in [0.01,0.1]$.  We
have used the same initial condition on $t \in [-2,0]$ for all
simulations.  Metastability is checked by tracking the zeroes of the
solutions. In the positive feedback case, we say that transient
oscillations end when the last pair of zeros of the solution
disappears.  In the negative feedback case, we say that transient
oscillations end when the solution has at most one pair of zeroes in
any interval of length two (called ``slow oscillations'').

In the following, in order to show whether the transient oscillations
time $T_\epsilon$ is of order $\exp(\frac{c}{\epsilon})$, we plot
$\epsilon$ Vs $\epsilon \log(T_\epsilon)$.  If $y(\epsilon) = \epsilon
\log(T_\epsilon)$ satisfies $y(0) > 0$, then  $T_\epsilon =
e^{\frac{c}{\epsilon}(1 + o(1))}$, meaning that oscillatory
transients are metastable.  On the other hand, if  $y(\epsilon) = \epsilon
\log(T_\epsilon)$ satisfies $y(0) = 0$, then  $T_\epsilon =
e^{o \left (\frac{1}{\epsilon} \right )}$, meaning that
oscillatory transients are not metastable.

\subsection{Case $\eta(0) \neq 0$} \label{etaneq0}

\par In this case, even when $\epsilon$ is very small, the
oscillations do not last for an exponentially long time. As discussed
in section~\ref{subsub_b}, the lack of metastability in our numerical
simulations follows from some of the results on a particular class of
state-dependent DDEs due to Mallet-Paret and Nussbaum \cite{MPN1,
  MPN2, MPN3}.  As $\epsilon$ tends to zero, numerical convergence to
the limit profile shape described by Mallet-Paret and Nussbaum is
very clear.

\begin{figure}[th!]
\begin{center}
\includegraphics[width=7.5cm]{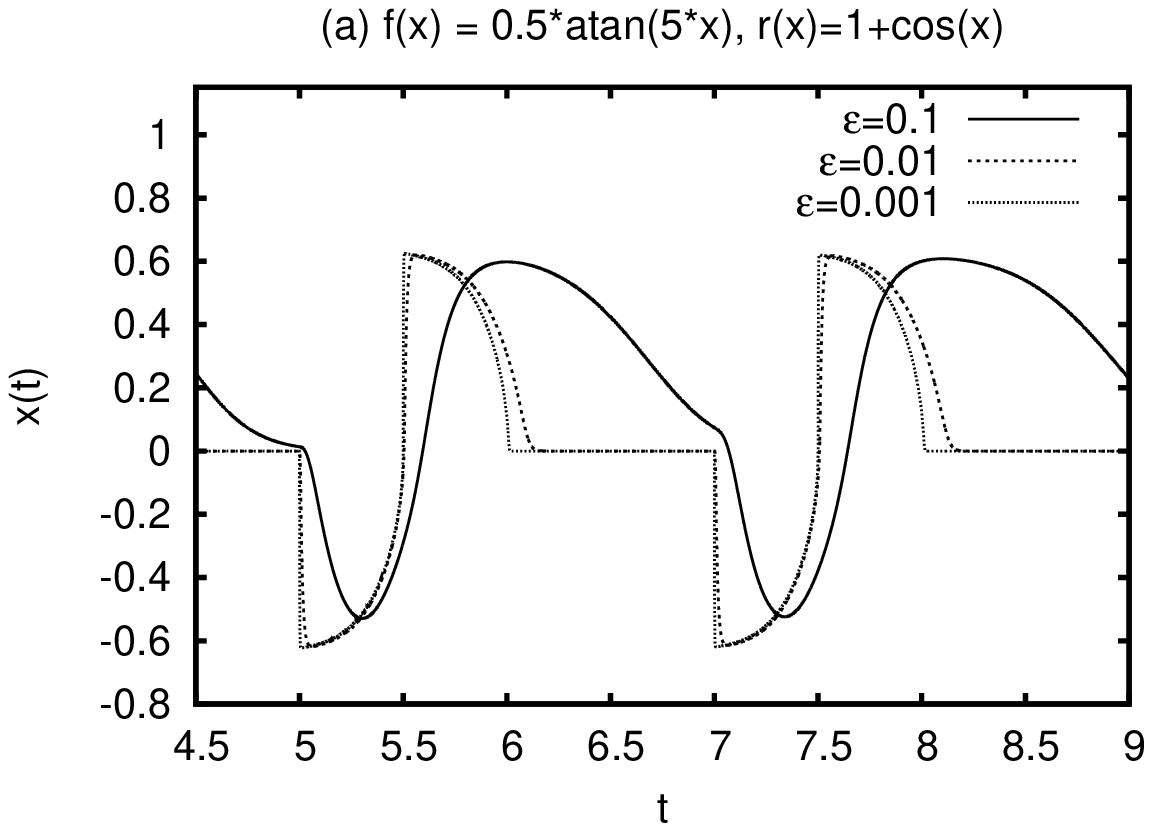} 
\includegraphics[width=7.5cm]{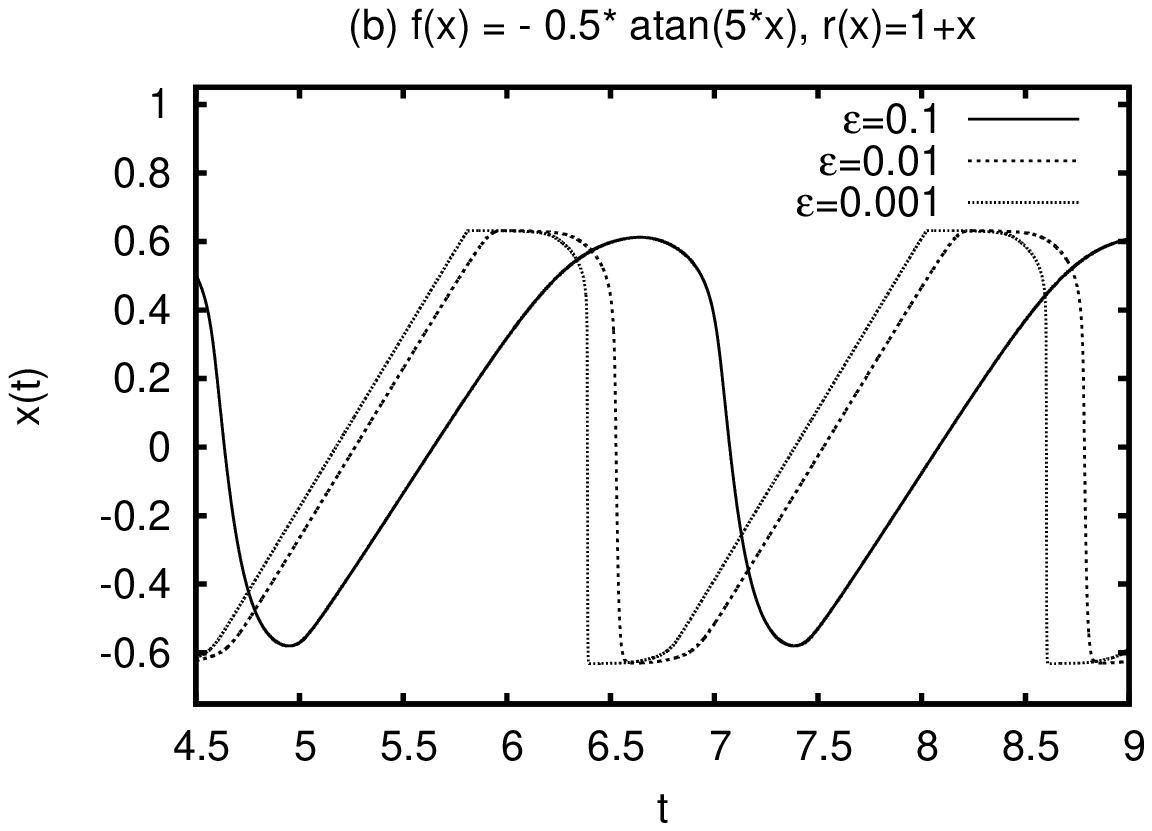}
\end{center}
\vspace{-1.0cm}
\begin{center}
\includegraphics[width=7.5cm]{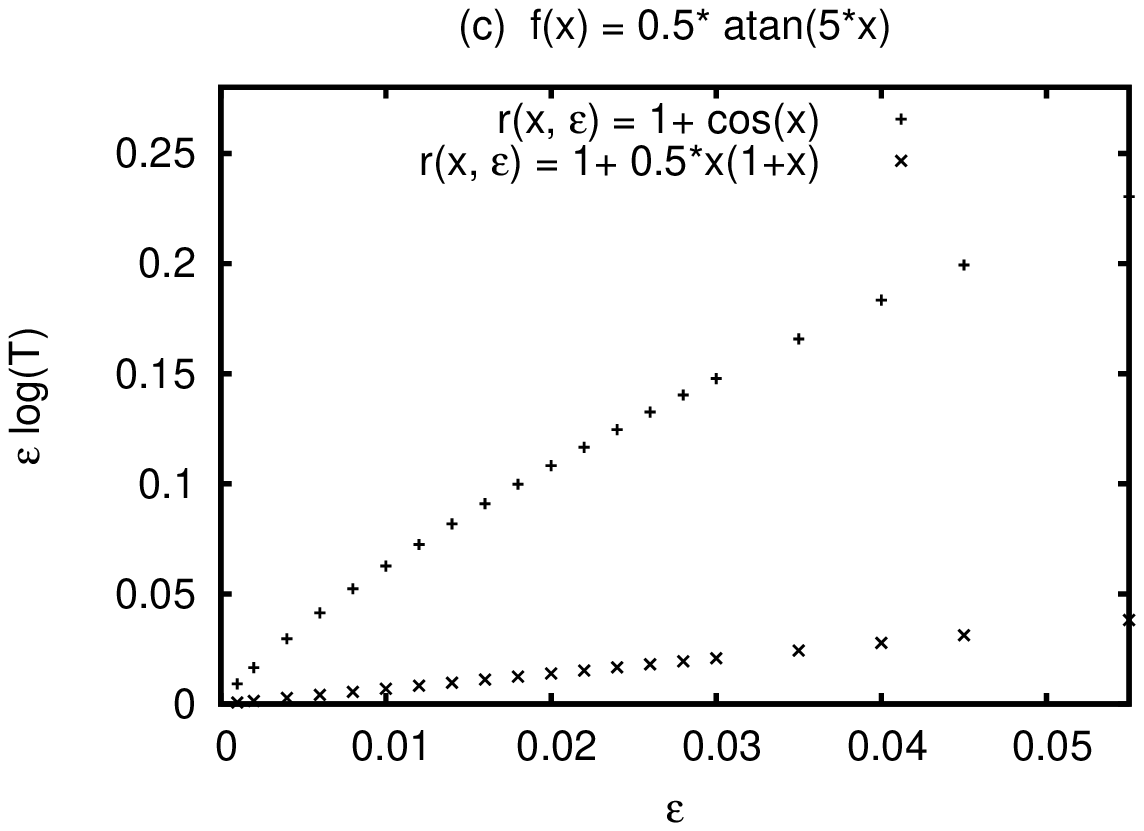} 
\includegraphics[width=7.5cm]{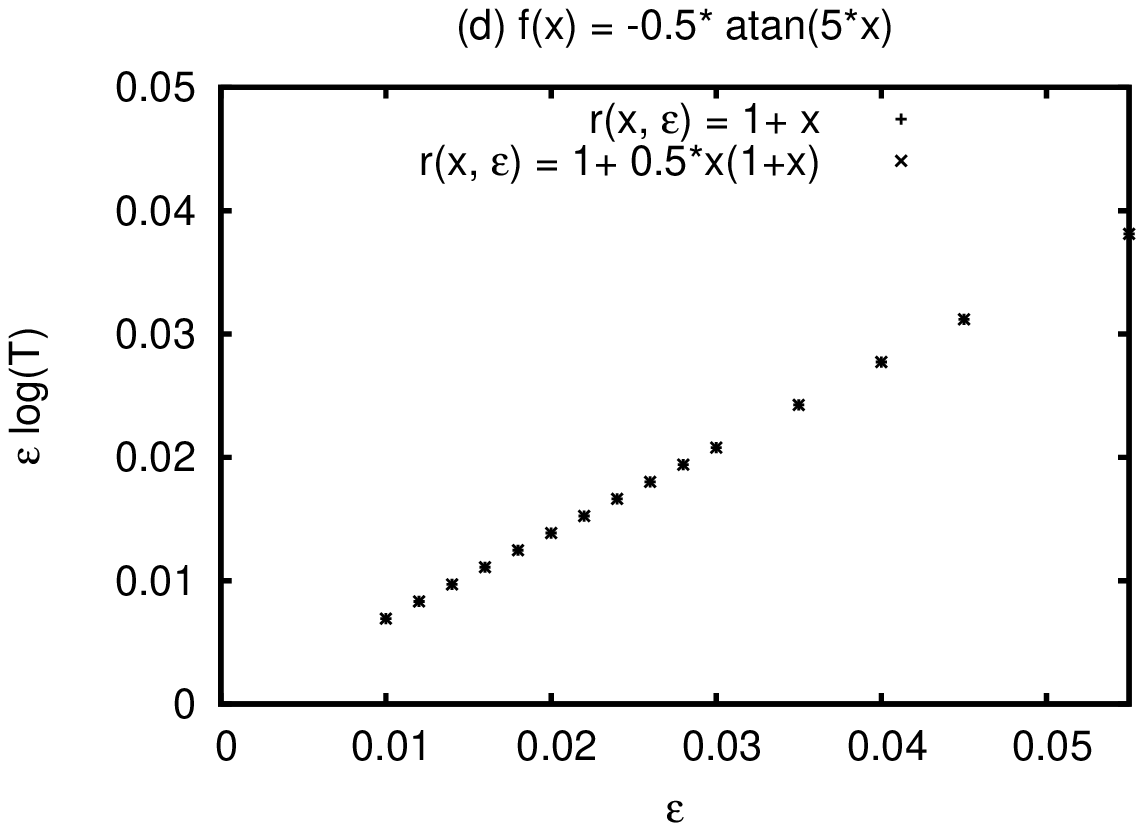} 
\end{center}
\vspace{-0.7cm}
\caption{\small{ Profile of solutions of equation (\ref{eqmain}) (top
    panels), and oscillatory transient duration (bottom panels), when
    $\eta(0) \neq 0$, with delay $r(x,\epsilon)= 1+R(x)$, for positive
    feedback $f(x) = \frac{1}{2} \arctan (5x)$ (panels on the left)
    and negative feedback $f(x) = - \frac{1}{2} \arctan (5x)$ (panels
    on the right). Top panels: solution profile for $\epsilon= 0.1,
    0.01, 0.001$, (a) positive feedback with $R(x)=\cos(x)$, (b)
    negative feedback with $R(x)=x$. (c) $\epsilon$
    Vs $\epsilon \log(T_\epsilon)$ for positive feedback, and delays
    $R(x)=\cos (x)$ (curve +), $R(x)=\frac{1}{2} x(1+x)$ (curve
    $\times$); (d) $\epsilon$ Vs $\epsilon \log(T_\epsilon)$ for
    negative feedback, and delays $R(x) = x$ (curve +),
    $R(x)=\frac{1}{2} x(1+x)$ (curve $\times$).  }}
\label{Fig_fronts_a0}
\end{figure}

\par The results are displayed in the figure~\ref{Fig_fronts_a0},
positive feedback on the left panels, negative feedback on the right
panels. The top panels in the figure~\ref{Fig_fronts_a0} exhibit the
solution profile for $\epsilon= 0.1, 0.01, 0.001$, and the bottom
panels display $\epsilon$ Vs $\epsilon \log(T_\epsilon)$.

The solution profile displayed in figure~\ref{Fig_fronts_a0}) (top
panels) show that when $\eta(0) \neq 0$, for both positive (top-left
panel) and negative (top-right panel) feedbacks, the oscillations are
not square-wave-like even for very small $\epsilon$, in contrast to
what happens for constant delay.

The bottom panels of figure~\ref{Fig_fronts_a0} show that when
$\eta(0) \neq 0$, oscillatory transient duration grows slowly when
$\epsilon$ converges to zero. For positive feedback
(figure~\ref{Fig_fronts_a0}(c)), when $R(x) =\frac{1}{2} x(1+x)$ (curve
$\times$), DITOs' duration never exceeds a few units of time for
$\epsilon > 0.001$; for $R(x)=\cos(x)$ (curve +) DITOs' duration tends
to $+ \infty$ but it does not grow as $e^{\frac c \epsilon}$ when
$\epsilon \rightarrow 0$ (the function $y(\epsilon) = \epsilon
\log(T)$ satisfies $y(\epsilon) \rightarrow 0$ when $\epsilon
\rightarrow 0$), and they are not metastable in this sense. For
negative feedback (figure~\ref{Fig_fronts_a0}(d)) the DITOs' duration
does not grow as $e^{\frac c \epsilon}$ when $\epsilon \rightarrow 0$,
meaning that DITOs are not metastable. Moreover, we can see that in
the negative feedback case the DITOs' duration depends very little on
$R(x)$, the curves $\epsilon$ Vs $\epsilon \log(T_\epsilon)$ being
almost identical for $R(x)=x$ and $R(x)=\frac{1}{2} x(1+x)$ (see
figure~\ref{Fig_fronts_a0}(d)).

\subsection{Case $\eta(0)=0$ with $0< \eta^\prime(0) < + \infty$ }\label{etaprimeinfty}

\par In this case, in addition to the Hopf bifurcation theorem,
Cooke's rescaling argument applies to Eq. (\ref{eqmain}) and one
expects that the rapidly oscillating solutions have large amplitude
when $\epsilon$ tends to zero.  We numerically observed that
metastable oscillations are almost $1+\epsilon \rho$ periodic (the
zeroes drift-speed $\rho$ depends on the feedback function $f$ and the
delay function $r$), and that the constants $\rho$ are coherent with
the corresponding constants in section \ref{subsection_tle_alpha=1}.
This period estimate is crucial for obtaining the transition layer
equation for the state dependent DDE (see section \ref{sec_tle},
equation (\ref{eqTLE}) ). Here we have taken $\eta(\epsilon) =
\epsilon$, so that $r(x, \epsilon) = 1 + \epsilon R(x)$. Metastability
is expected for any delay function $R$ in the negative feedback case
(symmetric and non-symmetric), while for the symmetric positive
feedback case the delay function $R$ must be even so as to guarantee
the symmetry condition.

\par Figure \ref{Fig_fronts_a1} displays the numerical results for
Eq. (\ref{eqmain}) when $\eta(\epsilon)=\epsilon $, for positive
feedback (panels on the left) and negative feedback (panels on the
right). The top and middle panels of figure~\ref{Fig_fronts_a1}
display the DITOs' profiles for $R(x)=\cos(x)$ and a
few $\epsilon$ values. The bottom panels of figure~\ref{Fig_fronts_a1}
display the transient duration ($\epsilon$ Vs $\epsilon
\log(T_\epsilon)$) for $R(x)=0$ (constant delay), $R(x)=\cos(x)$,
$R(x)=x$, $R(x)=$, $R(x)=\frac{1}{2} x(1+x)$, $R(x)= x²$.

The top panels of figure~\ref{Fig_fronts_a1} show that for both,
positive and negative feedback cases, the oscillatory solutions have a
square-wave-like shape when $\epsilon$ goes to zero, likewise the
constant delay case. Figure~\ref{Fig_fronts_a1}(a) for positive feedback
shows that, as $\epsilon \rightarrow 0$, the square-wave-like solution
has period $1 + c \epsilon$ (at first
order). Figure~\ref{Fig_fronts_a1}(b) for negative feedback shows that,
as $\epsilon \rightarrow 0$, the square-wave-like solution has period
$2 + c^\prime \epsilon$ (at first order).

\begin{figure}[th!]
\vspace{-0.2cm}
\begin{center}
\includegraphics[width=7.5cm]{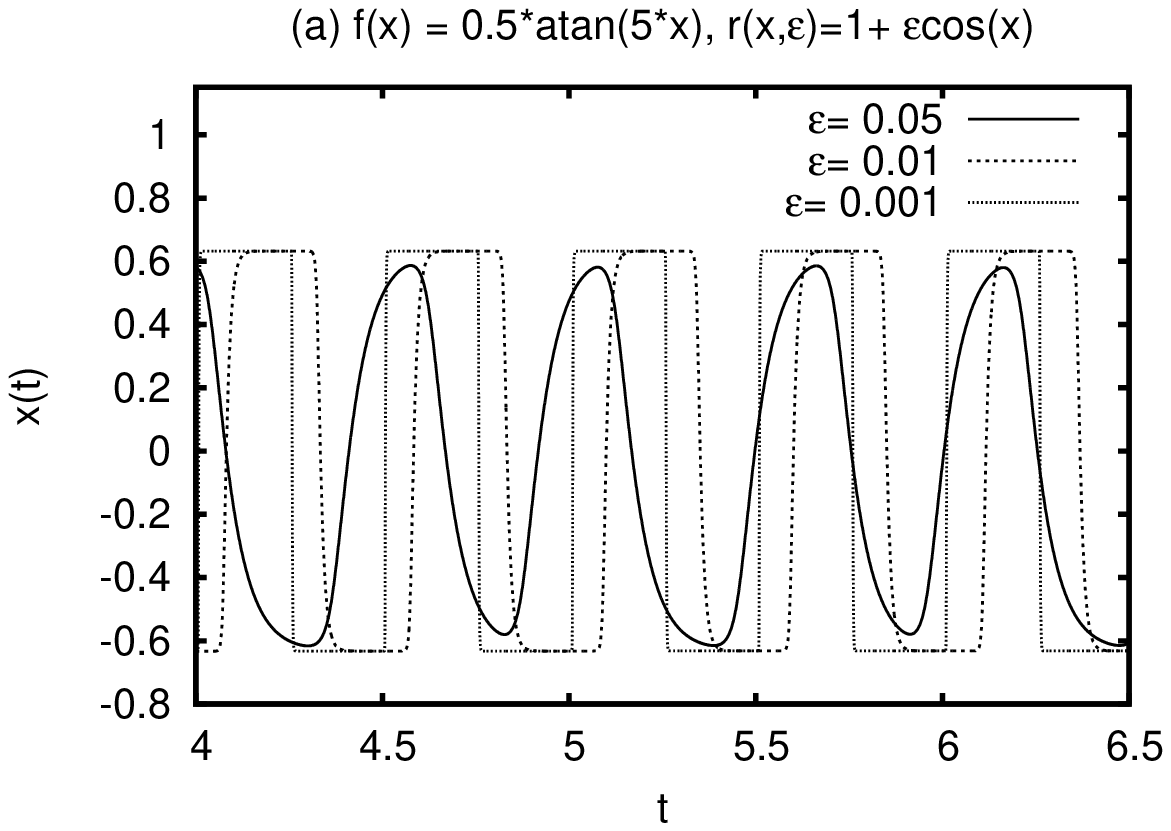} 
\includegraphics[width=7.5cm]{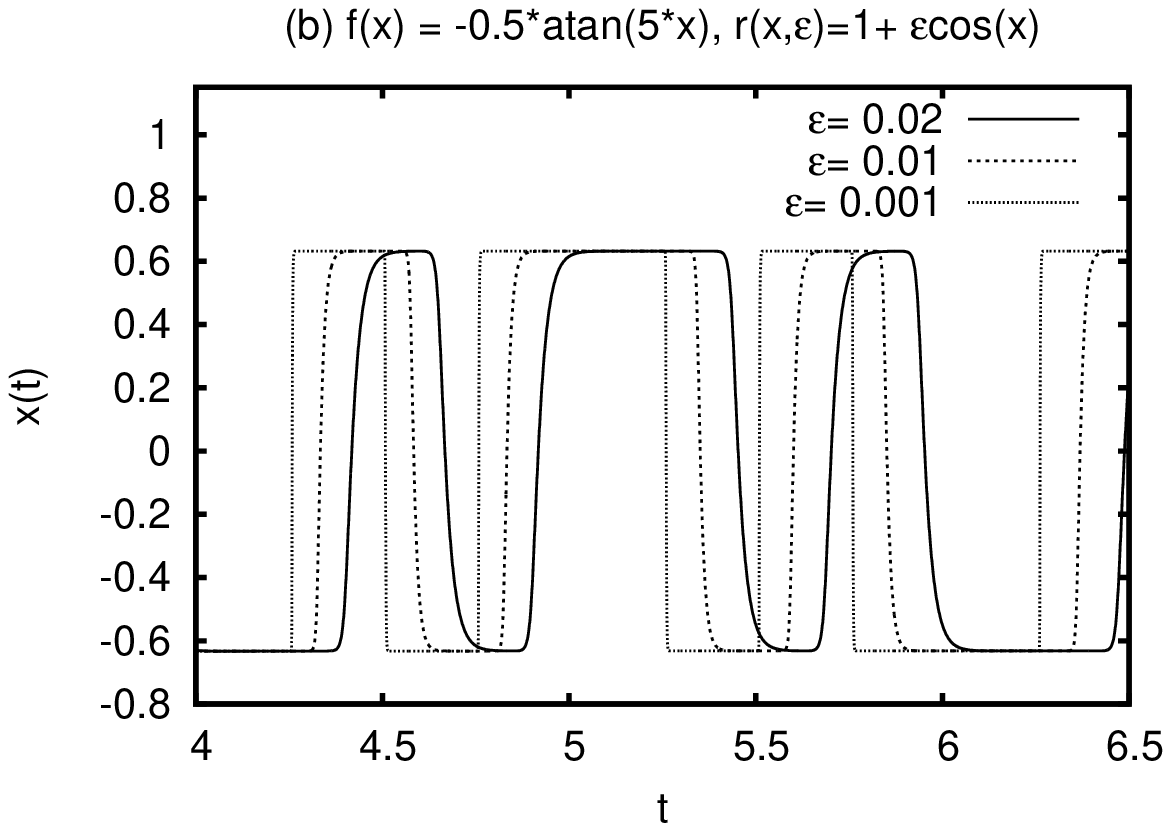} 
\end{center}
\vspace{-1.0cm}
\begin{center}
\includegraphics[width=7.5cm]{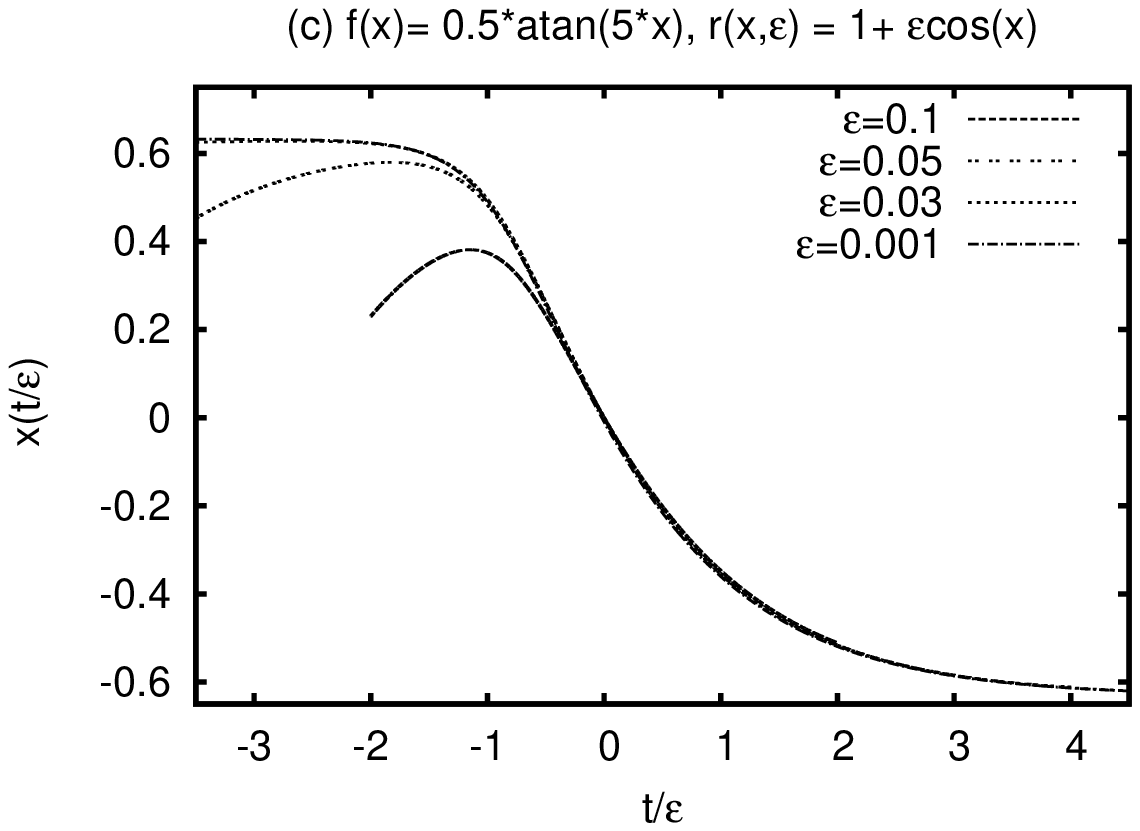} 
\includegraphics[width=7.5cm]{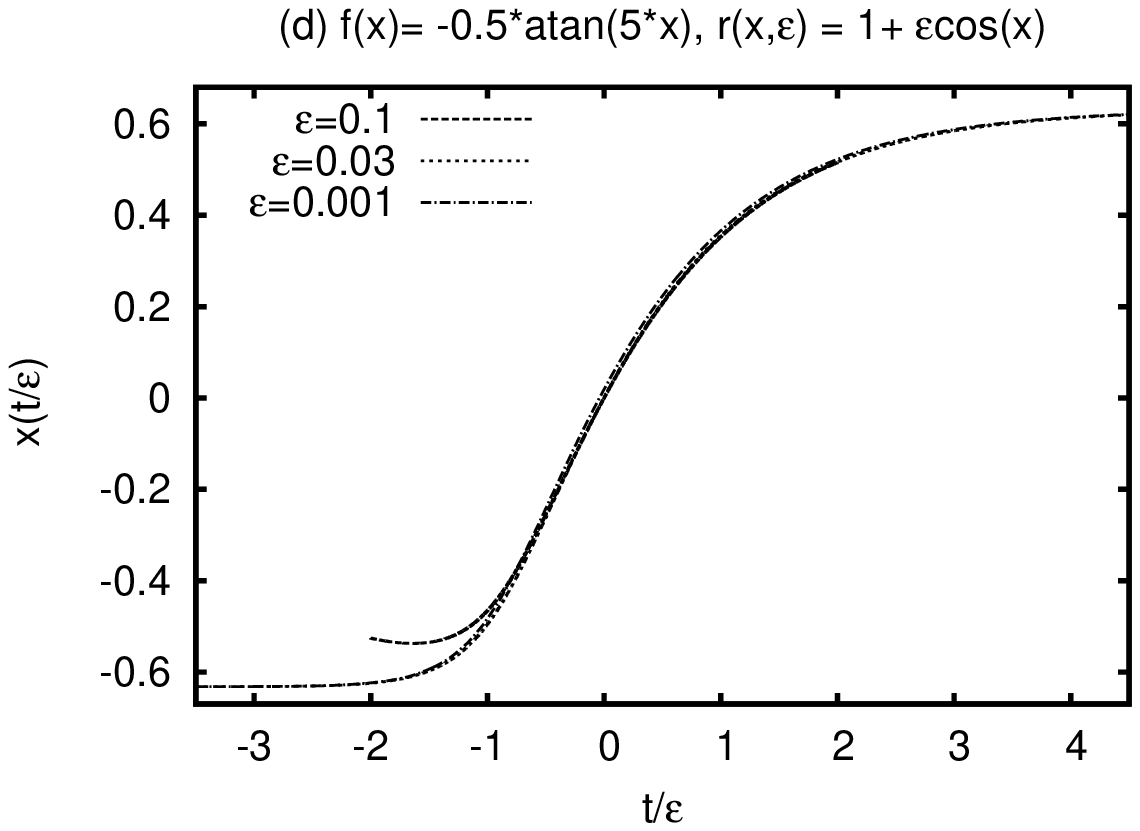} 
\end{center}
\vspace{-1.0cm}
\begin{center}
\includegraphics[width=7.5cm]{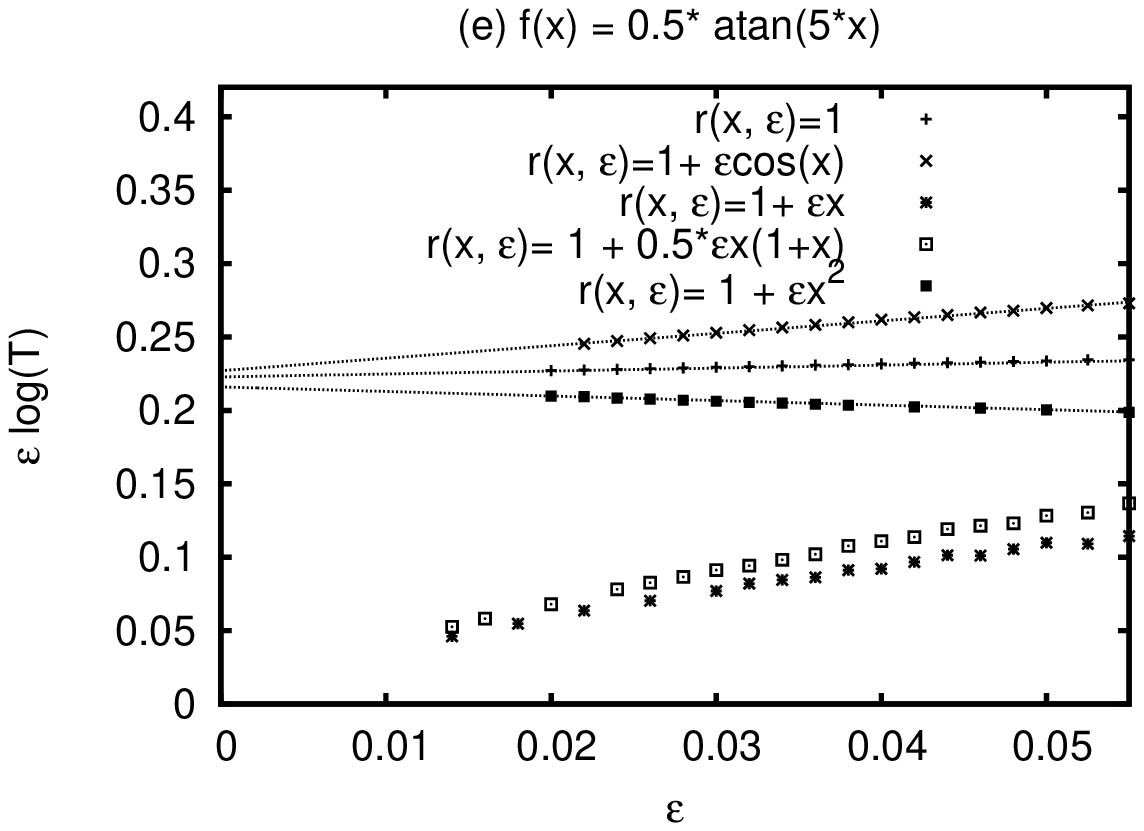} 
\includegraphics[width=7.5cm]{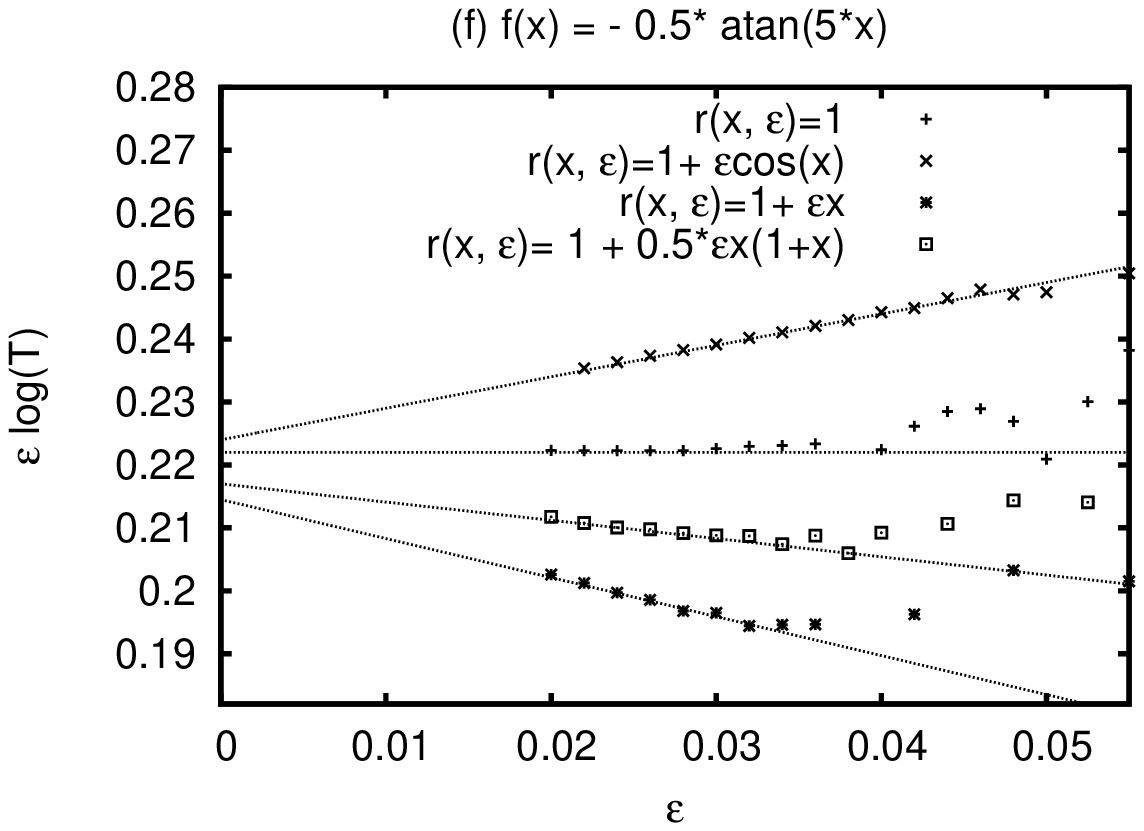} 
\end{center}
\vspace{-0.7cm}
\caption{{\small Profile of solutions of (\ref{eqmain}) (top and
    middle panels), and oscillatory transient duration (bottom
    panels), when $\eta(\epsilon)=\epsilon$, for positive feedback
    $f(x) = \frac 12 \arctan(5x)$ (panels on the left), and negative
    feedback $f(x) = - \frac 12 \arctan(5x)$ (panels on the right).
    {\bf Top panels:} solution profile for $R(x)=\cos (x)$ and
    decreasing $\epsilon$. {\bf Middle panels:} descending and
    ascending jumps obtained by the time scale $s=\frac{t}{\epsilon}$
    of the solutions profile for $R(x)=\cos (x)$, varying $\epsilon$:
    (c) positive feedback; (d) negative feedback. The curves for
    $\epsilon=0.03$ and $\epsilon=0.001$ are too close to be
    distinguished. {\bf Bottom panels:} (e) transient duration for
    positive feedback with delays $R(x)=0$, $R(x)=\cos (x)$, $R(x)=x$,
    $R(x)=\frac{1}{2} x(1+x)$, $R(x)=x²$; (f) transient duration for
    negative feedback with delays $R(x)=0$, $R(x)=\cos (x)$, $R(x)=x$,
    $R(x)=\frac{1}{2} x(1+x)$.  }}
\label{Fig_fronts_a1}
\end{figure}

\par The middle panels in figure \ref{Fig_fronts_a1} display the
square-wave-like oscillation after rescaling the time, $s=
\frac{t}{\epsilon}$. The time-rescaled profiles displayed in the
middle panels of figure \ref{Fig_fronts_a1} show convergence to a
limit profile when $\epsilon$ converges to zero, for both positive and
negative feedback, in agreement with the results in section
\ref{sec:periodic_sol}.  As explained in section
\ref{sec:periodic_sol}, after the appropriate time rescaling $s=
\frac{t}{\epsilon}$, the square wave shape converges to some limit
profile, which is solution of the transition layer problem
(\ref{eqTLE}).  Convergence to transition layer profiles occurs for
both increasing and decreasing rescaled jumps, for both positive and
negative feedbacks.

\par An approximation of the unknown constant $\rho= \rho^{\pm}$ in
problem (\ref{eqTLE}) can be obtained from the numerical approximate
period of metastable oscillations : $T \approx 1+ \epsilon \rho$.  We
have checked that the limit profiles agree
with the solutions of the transition layer equations obtained as
attractive fixed point of an appropriate operator $\mathcal T$, up to
numerical error of order $O(dt / \epsilon) $, where $dt$ is the
discretization parameter (see appendix section \ref{app_tlo}).
Due to regularity properties of the operator $\mathcal T$, this implies 
that the corresponding constants $\rho$ also agree at the same order.
 
\par The figure \ref{Fig_fronts_a1}f shows that
in the negative feedback case, metastable oscillatory patterns are
observed regardless of the delay function $R$.  For negative feedback
with constant delay ($R(x)=0$), the curve + in figure
\ref{Fig_fronts_a1}(f) shows that $T(\epsilon) \sim
\exp({\frac{c}{\epsilon}})$, for some non-zero constant $c$, as
expected. This same panel (f) shows that, for state dependent delay,
$\epsilon \log(T(\epsilon)) \approx c_{1} + c_{2} \epsilon +
o(\epsilon)$, implying that $T(\epsilon) =
\exp(\frac{c_{1}}{\epsilon}) \exp(c_{2}) (1 + o(1)) $.  The value of
constants $c_{1}$ and $c_{2}$ depends on the delay function $R$, but
in all cases $c_{1} > 0$ implying that the state-dependent DITOs' are
metastable.

\par State-dependent DITOs are metastable if and only if  the function
$y(\epsilon) = \epsilon \log(T)$ has a non-zero limit when $\epsilon
\rightarrow 0$.  Figure \ref{Fig_fronts_a1}e shows that in the
positive feedback case, the DITOs' are metastable only for odd
functions $f$ and even delay functions $R(x)$.  The curve + in the
figure \ref{Fig_fronts_a1}(e) shows that, for constant delay,
$T(\epsilon) \sim \exp({\frac{c}{\epsilon}})$, for some non-zero
constant $c$, as expected.  In the case of state dependent delay, when
the delay function is even, the curves for $R(x)=\cos(x)$, $R(x)=x^2$
in the figure \ref{Fig_fronts_a1}(e) show that $\epsilon
\log(T(\epsilon)) \approx c_{1} + c_{2} \epsilon + o(\epsilon)$,
meaning that $T = \exp(\frac{c_{1}}{\epsilon}) \exp(c_{2}) (1 + o(1))$
so that the state-dependent DITOs' are metastable.  In contrast, when
$R(x)$ is not even (cases $R(x) = \sin(x)$ , $R(x) = x + \frac{1}{2}
x^{2}$), figure \ref{Fig_fronts_a1}(e) shows that $\epsilon
\log(T(\epsilon)) \underset{\epsilon \rightarrow 0}{ \longrightarrow}
0$ so that $T(\epsilon)$ is not of order $\exp(\frac{c}{\epsilon})$,
implying that the state-dependent DITOs' are not metastable.

\subsection{Metastability induced by state dependent delay}
\label{metaasympf}

As already emphasized, for positive feedback $f$, DDE~(\ref{eqmain})
with constant delay exhibits metastability only if $f$ is odd.  The
constants $\rho^+$ and $\rho^-$ that solve the TLE problem are equal
in the case of odd positive feedback $f$ with constant delay, implying
metastability. When the positive feedback $f$ is not symmetric, the
condition $\rho^+ = \rho^-$ does not hold for constant delays, and
DITOs are not metastable. Nevertheless, in
subsection~\ref{subsec_meta} we have shown that for non-symmetric
positive feedback $f$, and state dependent delay $r(x,\epsilon)= 1 +
\lambda \epsilon R(x) $, for a critical value $\lambda=\lambda_c$, the
transition layer equations (\ref{eqTLElambda}) do have solutions such
$\rho^+=\rho^⁻$, provided $R(x)$ is not even (see
figure~\ref{Fig_rho_3}). Therefore, metastability has been induced by
introducing the appropriate state-dependence to DDE~(\ref{eqmain})
with non-symmetric positive feedback.  

In this section, we present the numerical solutions of equation
(\ref{eqmain}) for the same case analyzed in
section~\ref{subsec_meta}: $f(x)= \frac 12 \arctan(5(x+0.05)) -
\frac12 \arctan(0.25)$, with state-dependent delay $r(x,\epsilon)= 1 +
\lambda \epsilon R(x)$, $R(x)$ not even.  Panels (a) and (c) of
figure~\ref{fig:MetaS_sdd} display the results for the case
$R(x)=\frac 12 \epsilon x (1+x)$ (corresponding to
figure~\ref{Fig_rho_3}(d)); panels (b) and (d) display the results for
$R(x)=x$ (corresponding to figure~\ref{Fig_rho_3}(b)).

The fronts and the values of $\rho^+$ and $\rho^-$ change
continuously with $\lambda$.  In other words, simulations over fixed
durations of
DDE solutions for $\lambda$ in the vicinity of $\lambda_c$  are
similar because for $\lambda$ close to $\lambda_c$, the difference
$\rho^+-\rho^-$ is small and the DITOs are long lasting.  This
similarity not withstanding, the transient regime durations scale
differently with $\epsilon$ at $\lambda_c$ and nearby values
$\lambda$. Only at $\lambda_c$ the system becomes metastable in the
sense that the DITOs last for exponentially long times. This
difference is illustrated in the figures that show transient regime
duration at fixed $\lambda$ for various $\epsilon$ and the reverse,
i.e. at fixed $\epsilon$ for various values of $\lambda$.

The panels (a) and (b) of figure~\ref{fig:MetaS_sdd} display
$\epsilon$ Vs $1/\log(T(\epsilon)$, for three values of $\lambda$,
and we can see that when the parameter $\lambda$ is larger or smaller
than the critical value $\lambda_c$, the curves $y(\epsilon) = 1 /
\log (T_\epsilon)$ are very steep when $\epsilon \rightarrow 0$,
indicating that DITOs are not metastable in those cases. In contrast,
when $\lambda = \lambda_c$, the curves $y(\epsilon)$ have a bounded
slope as $\epsilon \rightarrow 0$, indicating that DITOs are
metastable. We have obtained $\lambda_c \approx 1.05$ for the case
$R(x)=\frac 12 \epsilon x (1+x)$ (figure~\ref{fig:MetaS_sdd}(a)), and
$\lambda_c \approx 0.55$ for the case $R(x)=x$
(figure~\ref{fig:MetaS_sdd}(b)). The difference between these values
of $\lambda_c$, and those values of $\lambda_c$ found in subsection
\ref{subsec_meta}, is smaller than $10^{-1}$, which is of the order of
numerical precision for these parameters.

The panels (c) and (d) of figure~\ref{fig:MetaS_sdd} display $\lambda$
Vs $\epsilon \log (T_{\epsilon})$, for three values of $\epsilon$, and
we can see that when $\epsilon$ is fixed, there is a unique value
$\lambda^\epsilon$ such that DITOs' duration is maximal, and it
converges to the critical value $\lambda_c$ as $\epsilon$ decreases.

The results in this subsection agree with the results in subsection
\ref{subsec_meta}, confirming the possibility of state-dependence of
the delay inducing metastability when the constant delay case does not
exhibit metastability. 

\begin{figure}[h!t]
\begin{center}
\includegraphics[width=7.5cm]{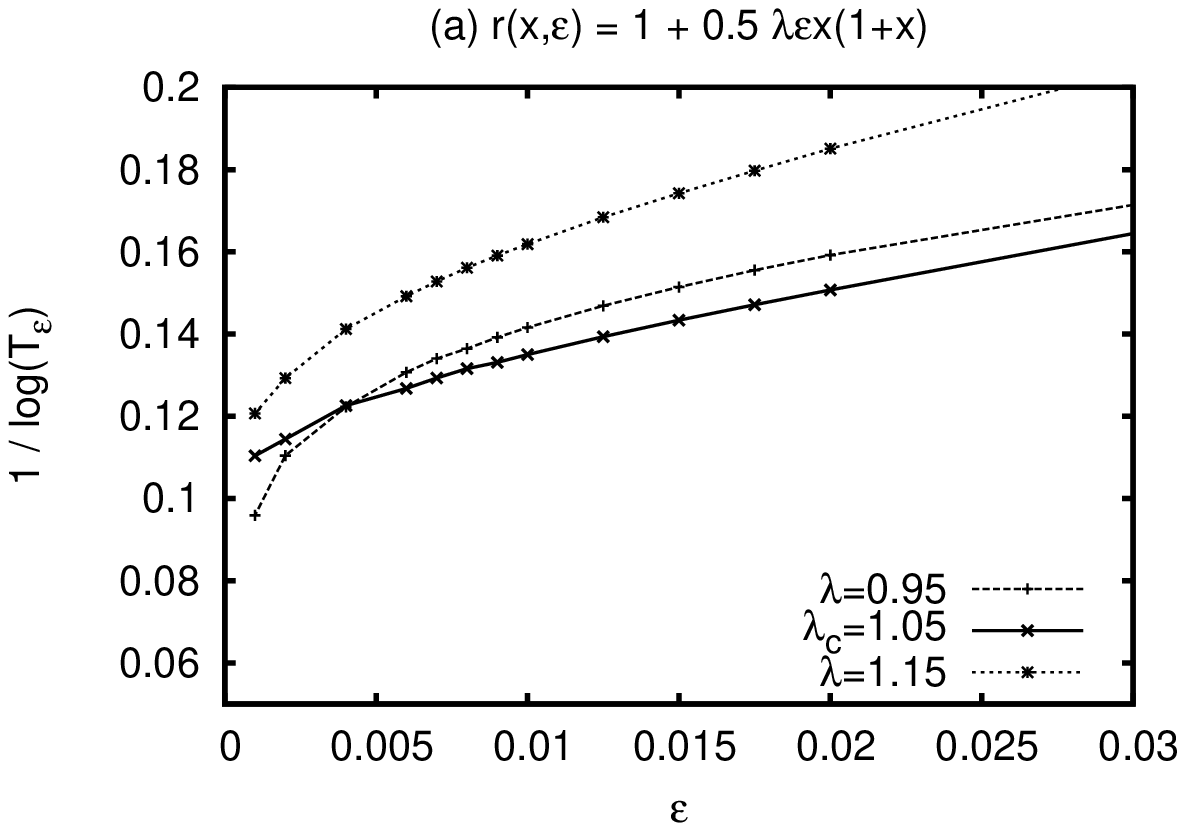}
\includegraphics[width=7.5cm]{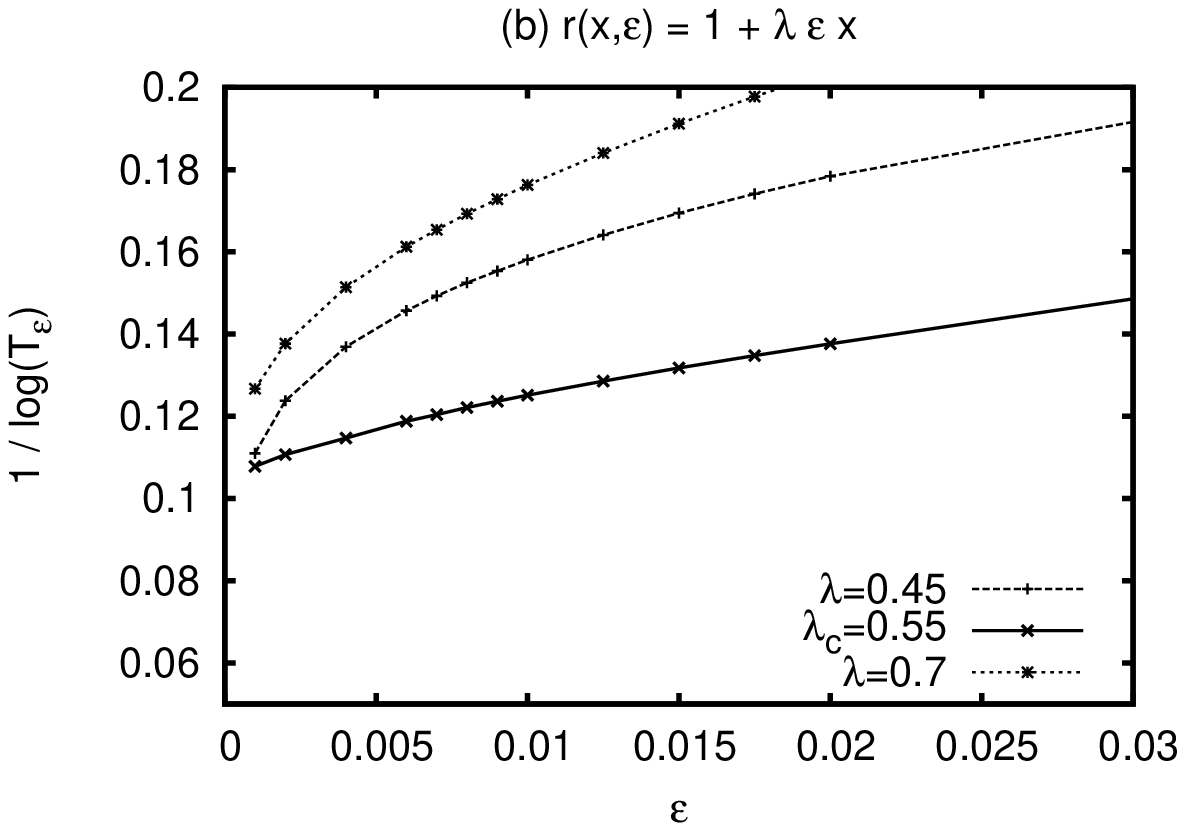}
\end{center}
\vspace{-0.9cm}
\begin{center}
\includegraphics[width=7.5cm]{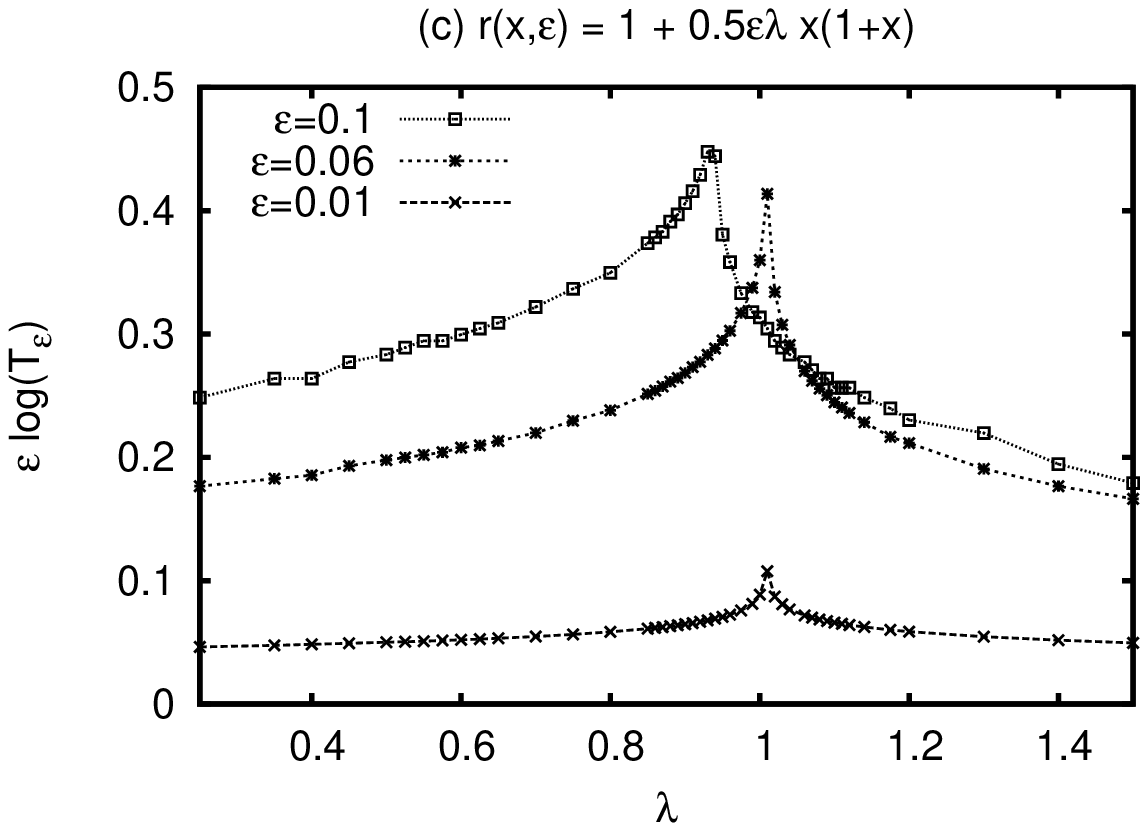}
\includegraphics[width=7.5cm]{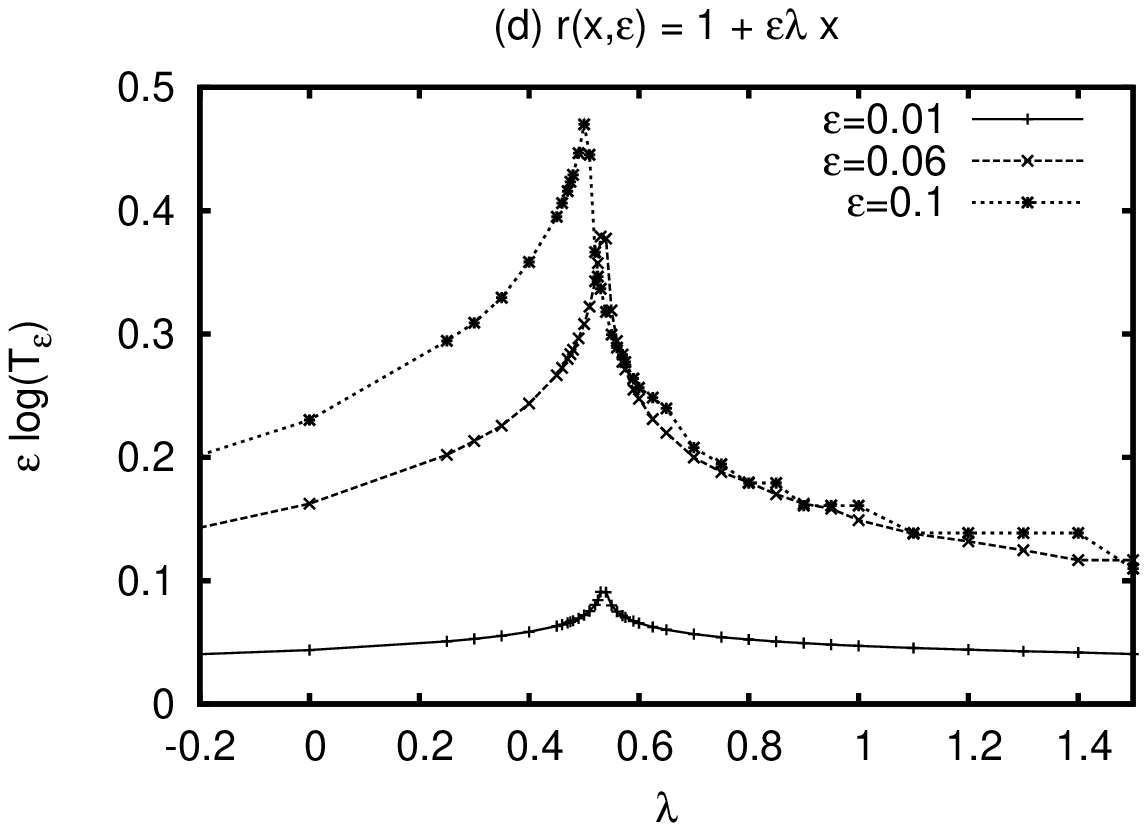}
\end{center}
\vspace{-0.7cm}
\caption{\small{ DITO duration for positive non-symmetric feedback
    $f(x) = \frac 12 \arctan(5(x+0.05)) - \frac 12 \arctan(0.25)$,
    with state dependent delay $r(x,\epsilon)= 1 + \lambda \epsilon
    R(x)$. Panels (a) and (c): $R(x)=  \frac 12 x(1+x)$;  panels (b)
    and (d):  $R(x)=x$.
}}
 \label{fig:MetaS_sdd}
\end{figure}

\subsection{Case $\eta(0)=0$ with $\eta^\prime(0) =0$ }
\label{etaprime0}
\par When $\eta^\prime(0)=0$, likewise the case $0 < \eta^\prime(0) < \infty $,
oscillations have a square wave shape, a ``period'' $T=1 + \epsilon
\rho$, and a rescaled limit transition-layer profile (figures not
shown).  The corresponding transition layer equation (\ref{eqTLEcst})
is independent of the function $\eta$, and it is the same transition
layer equation as for the constant-delay case $r(x,\epsilon)=1$.  

\par For given feedback $f$ and delay $R(x)$, if metastability occurs
when $0 < \eta^\prime(0) < + \infty$, our numerical investigation
indicates that it will also occur for $\eta^\prime(0)=0$. We have
computed the transient duration for the symmetric positive feedback
$f(x) = \frac 12 \arctan(5x)$, with delay $r(x, \epsilon) = 1 +
\eta(\epsilon) R(x)$, $\eta(\epsilon)= {\epsilon}^{\alpha}$, $\alpha>1$.

Figure~\ref{Fig_metaS_a2PF}~(a) displays $\epsilon$ Vs $\epsilon
  \log(T_\epsilon)$ for the state dependent delay $R(x)=\frac 12 x
  (1+x)$, and $\alpha= 1.5, 2, 3$. It shows that
  $\epsilon \log(T)$ converges to zero when $\epsilon \rightarrow 0$,
  regardless of $\alpha$, implying that for symmetric positive
  feedback the DITOs are not metastable if $R(x)$ is not even.
 
  Figure~\ref{Fig_metaS_a2PF}~(b) displays $\epsilon$ Vs $\epsilon
    \log(T_\epsilon)$ for the state dependent delay $R(x)=x^2$ (even
    function), and $\alpha= 2, 3$.  It shows that $\epsilon \log(T)$
    does not converge to zero when $\epsilon \rightarrow 0$, implying
    that for symmetric positive feedback the DITOs are metastable if
    $R(x)$ is even. From Figure~\ref{Fig_metaS_a2PF}~(b) we can also
      see that as $\ap$ increases, the duration $T$ of transient
      oscillations is closer and closer to the duration of the
      transient for the constant delay case.

For negative feedback the same result holds: metastable DITOs' for the
case $0 < \eta^\prime(0) < + \infty$ implies
 metastable DITOs' for the case $\eta^\prime(0)=0$ .

\begin{figure}[th!]
\begin{center}
\includegraphics[width=7.5cm]{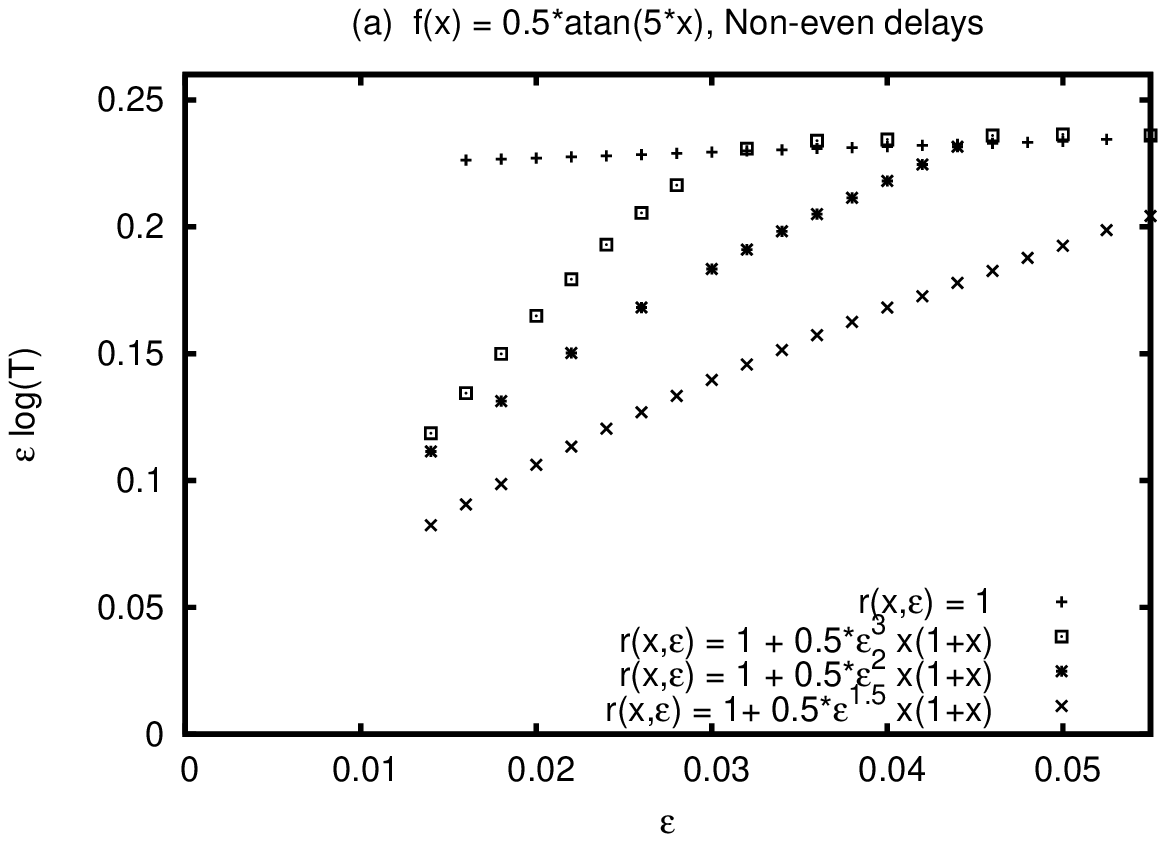} 
\includegraphics[width=7.5cm]{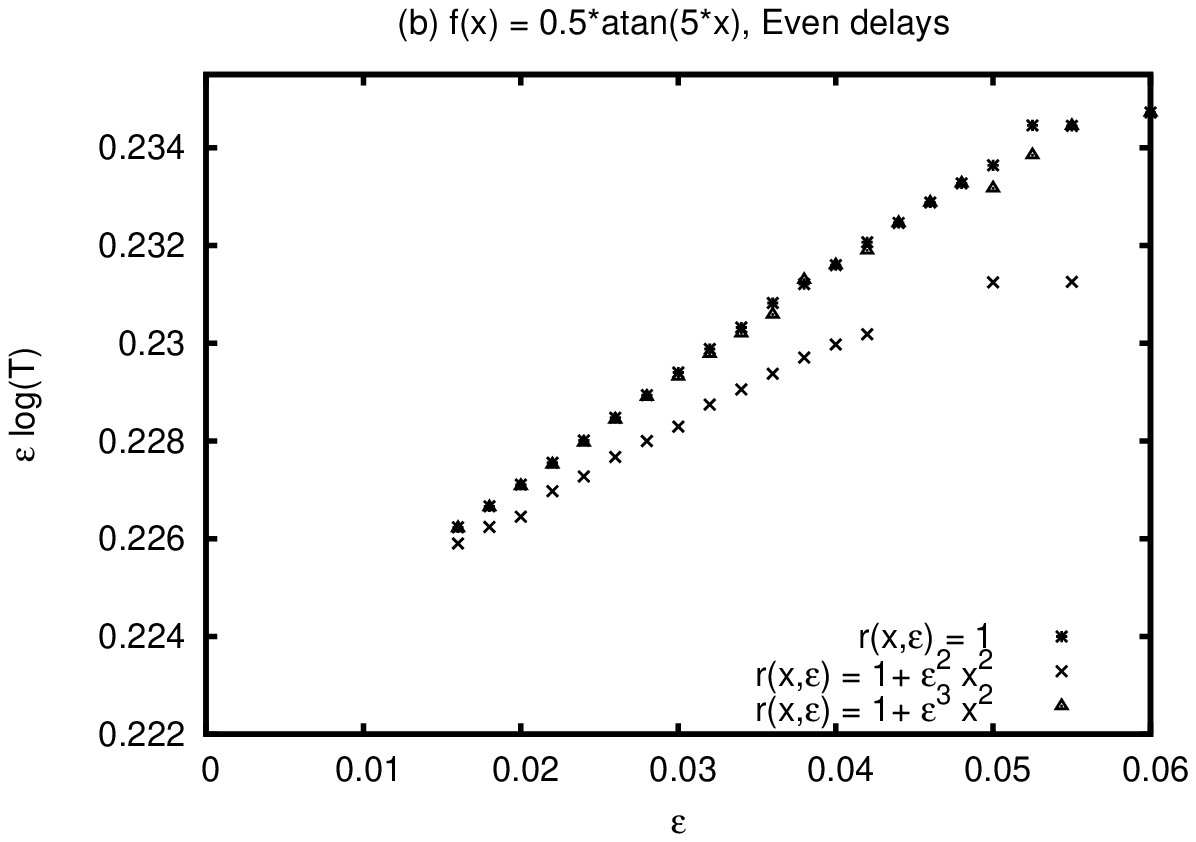} 
\end{center}
\vspace{-1.0cm}
\caption{\small{ Oscillatory transient duration for symmetric positive
    feedback $f(x) = \frac 12 \arctan(5x)$, with constant delay (curve
    +), and state dependent delay $r(x, \epsilon) = 1 + \eta(\epsilon)
    R(x)$, $\eta(\epsilon)=\epsilon^\alpha$, $\alpha > 1$: (a) $R(x)=
    \frac 12 x (1+x)$, $\alpha= 1.5, 2, 3$; (b) $R(x)= x^2$, $\alpha=
    2, 3$.  }}
\label{Fig_metaS_a2PF}
\end{figure}

\subsection{Case $\eta(0)= 0$ with $\eta^\prime(0)= + \infty$}

When $\eta(0)=0$, with $\eta^\prime(0) = + \infty$ oscillations have a
square-wave-like shape and an approximate period $T=1 +
\epsilon^{\alpha} \rho$ when $\epsilon \rightarrow 0$ (figures not
shown).  The usual time rescaling $s=\frac{t}{\epsilon}$ does not give
converging profiles, but the scaling $s=\frac{t}{\epsilon^{\alpha}}$
does.  As illustration, we display in figure \ref{Fig_front_a05} the
results for the negative feedback $f(x) = - \frac 12 \arctan(5x)$,
with $r = 1 + \epsilon^{1/2}\cos(x)$.  At the time scale $ t
\epsilon^{-1/2} $, we observe convergence of oscillating solutions
jumps to a limit transition layer profile when $\epsilon \rightarrow
0$.

\begin{figure}[th!]
\begin{center}
\includegraphics[width=7.5cm]{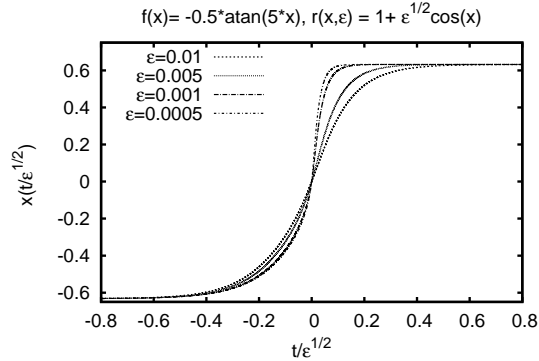} 
\end{center}
\vspace{-1.0cm}
\caption{\small{ Convergence of rescaled jumps of oscillating
    solutions in the limit $\epsilon \rightarrow 0$, when $\eta(0)=0$
    with $\eta^\prime(0) = + \infty$.  We used negative feedback $f(x)
    = - \frac{1}{2} \arctan(5x)$, with $r = 1 + \epsilon^{1/2}\cos(x)$.}}
 \label{Fig_front_a05}
\end{figure}

\par Here neither Mallet-Paret and Nussbaum's nor Cooke's argument
apply, so that we have no indication on whether rapidly oscillating
periodic solutions (which are expected because of Eichmann's Hopf
bifurcation theorem) have large or small amplitudes.  Numerical
simulations have shown that oscillatory transients can last for a very
long time, but due to numerical difficulties it is not clear whether
this transient duration is exponential, hence we cannot say whether
these long lasting oscillatory transients are indeed metastable or
not.

\section{Discussion and conclusion}\label{sec_conclusion}

\par At present time, no global geometric
characterization of the organization of the phase portrait of
scalar state dependent DDEs is at hand.  Even basic results, such as
the Hopf bifurcation theorem, have only recently been established.
However, based on current knowledge and systematic numerical explorations, 
it is possible as done in section \ref{sec:periodic_sol} to conjecture that 
the phase portraits of DDEs with state dependent delays have the same geometrical
organization as those of constant delays.  
Furthermore, it is possible to discuss the putative occurrence of
metastability based upon  available information regarding rapidly oscillating
periodic solutions, and educated guesses supported by numerical
investigations.

\par In this paper we have shown
that metastable oscillating solutions can exist in singularly
perturbed DDEs with state dependent delays of type (\ref{eqmain}). Based both on
mathematical and numerical results, we have been able to link the
properties of Eq.  (\ref{eqmain}), and its solutions, to the occurrence
of oscillatory metastable transients.

\par Such metastable transients have never been observed when the delay
$r=r(x)$ does not depend on the singular parameter $\epsilon$. The main reason
for the short duration of such DITOs'  is the
decreasing size of rapidly oscillating periodic oscillations. 
Considering delays of the form
$r(x,\epsilon)=1 + \eta(\epsilon)R(x)$, we found that the scaling
parameter value $\eta(\epsilon)$ is crucial for the existence of
metastable transients.  When $\eta(0)\neq 0$ exponentially long
lasting oscillatory transients were never observed, even for very
small $\epsilon$.  When $\eta(0)=0$ and $\eta^\prime(0) = \infty $
long lasting oscillatory transients can be observed, but due to
numerical difficulties one cannot conclude whether these transients
are indeed metastable or not.  For $\eta(\epsilon) = \epsilon^\alpha$,
the larger $\alpha$ is, the longer transient oscillations will last.
When $\eta(0)=0$ and $0 \le \eta^\prime(0) < \infty $ metastable
oscillatory transients can always be observed for negative feedback,
while for positive feedback  some symmetry condition must be satisfied
by  the feedback function $f$ and the delay function $R$.

\par There are two main tools to analyze metastability phenomena.  With a
geometric approach, one can look for the existence of a global
attractor containing a ``cascade'' of unstable periodic orbits and
heteroclinic connections between them; and with a dynamical approach
one can investigate transition layer equations that describe the
asymptotic shape of the oscillations when $\epsilon$ converges to
zero.

\par We have shown that, with few hypotheses on the functions $f$ and
$r$, a Hopf bifurcation theorem of Eichmann applies to Eq.
(\ref{eqmain}). This implies the existence of a sequence of Hopf
bifurcations as $\epsilon$ converges to zero, meaning that a global
attractor with the previously described structure might exist in many
cases. Nevertheless, metastable oscillations cannot be observed in general.
Our numerical investigation has shown that when the delay function $r$ does not
converge to a constant as $\epsilon$ converges to zero (see
figure~\ref{Fig_fronts_a0} in subsection~\ref{etaneq0}),
the amplitude of periodic solutions has to
converge to zero as the period converge to zero, and metastable
oscillations were not observed.  This suggests that not only rapidly
oscillating but also large amplitude periodic solutions are needed to
support metastability.

\par Furthermore, even when the delay function $r$ converges to a
constant as $\epsilon$ tends to zero (section~\ref{section_num}, cases
$\eta(0)=0$), the state dependent DITOs' are metastable only in those
cases where there exist transition layer equations similar to those
found for DDEs of type (\ref{eqmain}) with constant delay ($\alpha\ge
1$, subsections~\ref{etaprime0}, and \ref{etaprimeinfty}).  This
suggests that metastable state dependent DITOs' cannot exist unless
the oscillations have a limiting shape that are solution to a
transition layer problem of the form (\ref{eqTLE}), as $\epsilon$
converges to zero.

\par Our numerical analysis of the transient oscillations in DDE
(\ref{eqmain}) induced by state dependent delay of the form $r(x)= 1 +
\eta(\epsilon) R(x)$, has revealed that for negative feedback $f$,
metastable state dependent DITOs' exist in the same way as for the
constant delay case, while for positive feedback $f$, metastable state
dependent DITOs' exist if and only if $f$ and $R$ satisfy some
symmetry conditions. If the positive feedback $f$ is an odd function
(symmetry requirement for metastability in the case of constant
delay), then metastable state dependent DITOs' exist only if the delay
$R(x)$ is an even function.  We also have showed that by adding state
dependence to the delay, it is possible to obtain metastability for
positive feedback $f$ for which the constant delay transient
oscillations are not metastable (see subsections~\ref{subsec_meta} and
\ref{metaasympf}).

\par An important contribution of this work has been the introduction
of a novel class of transition layer equations associated with state
dependent delays. From our numerical investigations, we claim that
these equations capture two essential aspects of the dynamics of DDEs
with state dependent delays. The first is the shape of the
oscillations as the parameter $\ep$ becomes small. The second is the
drift of the oscillations. Our focus has been on monotone feedbacks,
nevertheless the transition layer equations we have introduced remain
valid for non monotone feedbacks as well.  Our paper paves the way to investigate
 the dynamics of DDEs with such feedbacks and
state dependent delays, through the novel transition layer equations.

\appendix

\section{Proof of theorem \ref{thm_hopf} }
\label{subsec_hopf}

\par The proof of theorem \ref{thm_hopf} is based on a ``Hopf
bifurcation theorem'' proven by Eichmann (Eichmann's PhD thesis
\cite{eichmann} p 81).  For $\epsilon \in ]0; 1[$ and $y \in
C^{1}([-M,0],\mathbb{R})$, let $g$ be the function $g(\epsilon,
y)=\frac{1}{\epsilon}\left ( - y(0) + f( y(-r(\epsilon, y)) ) \right
)$, such that Eq.  (\ref{eqmain}) is equivalent to
\begin{equation}\label{eq:eichmann}
x^{\prime}(t)=g(\epsilon, x_t),
\end{equation}
with the notation $x_t(s) = x(t+s)$ for all $s \in [-M,0]$.

\par The Hopf bifurcation theorem of Eichmann has three first order
derivatives hypotheses $(H_1)$ $(H_2)$ and $(H_3)$, three second order
derivatives hypotheses $(H_4)$ $(H_5)$ $(H_6)$ and three spectral
hypothesis $(L_1)$, $(L_2)$ and $(L_3)$.  For convenience of the
reader we state these hypotheses here before proving that they are
satisfied by equations (\ref{eqmain}) and (\ref{eq_delay_eta}).

Suppose that there is an open subset $U$ of $C^1([-M, 0], \mathbb R)$ such that 
\begin{itemize} 
\item $(H_1)$ the mapping $g : \, ]0, 1[ \times U \rightarrow \mathbb
  R$ is continuously differentiable,
\item $(H_2)$ for any $(\epsilon, x) \in ]0, 1[ \times U $ the second
  partial derivative $D_x g (\epsilon, x)$ can be extended to a linear
  continuous map
\begin{equation*}
D_x g (\epsilon, x) : \, C^0([-M,0], \mathbb R) \rightarrow \mathbb R
\; ;
\end{equation*}
\item $(H_3)$ the (extended) mapping 
\begin{equation*}
\begin{array}{cl}
]0, 1[ \times U \times C^0([-M,0], \mathbb R) & \rightarrow \mathbb R \\
(\epsilon, x, y) & \mapsto D_x g (\epsilon, x) y, 
\end{array}
\end{equation*} is continuous;
\item $(H_4)$ the mapping $g : \, ]0, 1[ \times \left ( U \cap C^2
  \right ) \rightarrow \mathbb R$ is twice continuously differentiable;
\item $(H_5)$ for any $(\epsilon, x) \in ]0, 1[ \times U \cap C^2 $
  the second order partial derivative $D^2_x g (\epsilon, x)$ can be
  extended to a bilinear continuous map
\begin{equation*}
  D^2_x g (\epsilon, x) : \,  C^1([-M,0], \mathbb R) \times
  C^1([-M,0], \mathbb R) \rightarrow \mathbb R \; ;
\end{equation*}
\item $(H_6)$ the (extended) mapping
\begin{equation*}
\begin{array}{cl}
]0, 1[ \times U \times C^1([-M,0], \mathbb R) \times C^1([-M,0], \mathbb R) & \rightarrow \mathbb R \\
(\epsilon, x, y, z) & \mapsto D^2_x g (\epsilon, x) (y,z) 
\end{array}
\end{equation*} is continuous, and 
\begin{equation*}
\begin{array}{cl}
]0, 1[ \times U \times C^1([-M,0], \mathbb R) & \rightarrow \mathcal L( C^2, \mathbb R) \\
(\epsilon, x, y) & \mapsto D^2_x g (\epsilon, x) (y, \cdot), 
\end{array}
\end{equation*} is continuous (where $\mathcal L( C^2, \mathbb R)$ is the space of linear functionals from $C^2([-M,0],\mathbb R)$ to $\mathbb R$).
\end{itemize}

Suppose that $g(\epsilon, 0)=0$, for any $\epsilon \in ]0, 1[$, 
and let $A(\epsilon)$ be the generator of the strongly continuous semigroup 
on $C^0([-M,0], \mathbb R)$ generated by the linearized equation
\begin{equation*}
y_t^\prime = D_x g(\epsilon, 0) y_t.
\end{equation*}

Suppose that there is a $\epsilon^* \in ]0,1[$ and some open interval
$]\epsilon^* - \eta, \epsilon^* + \eta[ \subset ]0,1[$ such that
\begin{itemize}
\item $(L_1)$ for any $\epsilon \in ]\epsilon^* - \eta, \epsilon^* +
  \eta[$, there is a simple eigenvalue $\lambda(\epsilon)$ of
  $A(\epsilon)$, such that the mapping $ \epsilon \mapsto
  \lambda(\epsilon) $ is $C^{1}(]\epsilon^* - \eta, \epsilon^* +
  \eta[, \mathbb C)$,
\item $(L_2)$ the eigenvalue $\lambda(\epsilon)$ crosses the imaginary
  axis at $\epsilon^*$ : $\Re ( \lambda (\epsilon^*)) = 0$ and $\Im (
  \lambda (\epsilon^*)) = \omega_0 > 0$, and $\frac{d
    \lambda}{d\epsilon} (\epsilon^*) \neq 0$,
\item $(L_3)$ and for any $k \in \mathbb Z - \{-1,1\}$, $\nu = i k
  \omega_0 $ is not an eigenvalue of $A(\epsilon^*)$.
\end{itemize}

Now we prove that theorem \ref{thm_hopf} is a consequence of
Eichmann's theorem. To this end we start by proving the following
Propositions \ref{prop1} and \ref{prop2}.

\begin{proposition}\label{prop1}
  Let $U$ be an open subset of $C^{1}([-M,0], \mathbb{R})$.  Suppose
  $f : \mathbb{R} \rightarrow \mathbb{R}$ is $C^1$ and $r : ]0; 1[
  \times U \rightarrow [0,M]$ is $C^1$.  Suppose that for any
  $(\epsilon, x_t) \in ]0;1[ \times U$, $\frac{\partial r}{ \partial
    y} : ]0;1[ \times C^{1}([-M,0], \mathbb{R}) \rightarrow \mathbb{R}
  $ can be extended to a continuous linear map $]0;1[ \times
  C^{0}([-M,0], \mathbb{R}) \rightarrow \mathbb{R}$, and that $
  (\epsilon, x_t, h_t) \rightarrow \frac{\partial r}{ \partial
    y}(\epsilon, x_t).h_t $ is continuous $]0;1[ \times U \times
  C^{0}([-M,0], \mathbb{R}) \rightarrow \mathbb{R}$.

Then the regularity hypotheses $H_1$, $H_2$ and $H_3$ of Eichmann's 
Hopf Bifurcation Theorem hold.
\end{proposition}
Here, $\frac{\partial}{\partial \epsilon}$ and
$\frac{\partial}{\partial y}$ denote Frechet derivatives.
\begin{proof}
One finds 
$$\frac{\partial g}{\partial \epsilon} (\epsilon, y) = \frac{-1}{\epsilon}g(\epsilon, y)  
- \frac{1}{\epsilon} \left ( \frac{\partial r}{ \partial
    \epsilon}(\epsilon, y) y^{\prime}(-r(\epsilon, y)) f^{\prime}(y(-
  r(\epsilon,y))) \right ) $$
$$\frac{\partial g}{\partial y}(\epsilon, x).h = 
\frac{1}{\epsilon} \left [ - h(0) + \left (- x^{\prime}(- r(\epsilon,
    x)) \frac{\partial r}{\partial y}(\epsilon, x) .h + h(-
    r(\epsilon, x)) \right ) f^{\prime}(x(- r(\epsilon, x))) \right ]
$$

This shows that $\frac{\partial g}{\partial \epsilon} : ]0,1[ \times U
\rightarrow \mathbb{R}$ and $\frac{\partial g}{\partial y} : ]0,1[
\times U \rightarrow \mathbb{R}$ are $C^0$, so that $g \in C^1$, and
$H_1$ holds.  Using the hypotheses made on $r$ in the theorem
\ref{thm_hopf}, the formula shows that $H_2$ and $H_3$ also hold.
\end{proof}

\begin{proposition}\label{prop2}
  Suppose that the hypotheses of Proposition \ref{prop1} holds.
  Suppose additionally that $f : \mathbb{R} \rightarrow \mathbb{R}$ is
  $C^2$ and that $ r : \, ]0,1[ \times \left ( U \cap C^2([-M,0],
    \mathbb{R}) \right ) \rightarrow [0,M] \subset \mathbb{R} $ is
  $C^2$.  Suppose that for any $(\epsilon, x) \in ]0,1[ \times U$, the
  function $\frac{\partial^2 r}{\partial y^2} (\epsilon, x)$ has a
  bilinear continuous extension to $C^{1}([-M,0],\mathbb{R}) \times
  C^{1}([-M,0],\mathbb{R}) \rightarrow \mathbb{R}$, which depends
  continuously on $(\epsilon, x) \in ]0,1[ \times (U \cap C^2 )$.  And
  suppose that $ (\epsilon, x, h) \mapsto \frac{\partial^2 r}{\partial
    y^2}(\epsilon, x)(h, \cdot) $ is continuous $ ]0,1[ \times (U \cap
  C^2) \times C^1 \rightarrow \mathcal{L}(C^2, \mathbb{R}) $.

  Then regularity hypotheses $H_4$, $H_5$ and $H_6$ of Eichmann's Hopf
  Bifurcation Theorem are true.
\end{proposition}
\begin{proof}
  This proposition is a consequence of the following formulas for second
  order derivatives of function $g$.
$$- \frac{\partial }{\partial \epsilon} \left ( \epsilon \frac{\partial g}{ \partial \epsilon} + g \right )(\epsilon ,x) = 
\frac{\partial^2 r}{\partial \epsilon^2}(\epsilon ,x) x^\prime(-
r(\epsilon ,x)) f^\prime(x(- r(\epsilon ,x))) $$
$$ 
- x^{\prime \prime}(-r(\epsilon ,x)) \left ( \frac{\partial
    r}{\partial \epsilon}(\epsilon ,x) \right )^2
f^{\prime}(x(-r(\epsilon ,x))) - \left ( x^\prime(- r(\epsilon ,x))
  \frac{\partial r}{\partial \epsilon}(\epsilon ,x) \right )^2
f^{\prime \prime}(x(-r(\epsilon ,x)))
 $$
and 
$$ - \frac{\partial}{\partial y}  \left ( \epsilon \frac{\partial g}{\partial \epsilon} + g \right ) (\epsilon,x).h = 
\left ( \frac{\partial^2 r}{\partial y \partial \epsilon}(\epsilon, x).h \right )
x^{\prime}(- r(\epsilon,x)) f^\prime(x(-r(\epsilon,x)))$$
$$- \frac{\partial r}{\partial \epsilon}(\epsilon,x) x^{\prime \prime}(- r(\epsilon,x)) 
\left ( \frac{\partial r}{\partial y}(\epsilon,x).h \right ) f^\prime(x(-r(\epsilon,x)))$$
$$+\frac{\partial r}{\partial \epsilon}(\epsilon,x) x^{\prime}(- r(\epsilon,x)) f^{\prime \prime}(x(-r(\epsilon,x))) 
\left [ - x^\prime(-r(\epsilon,x)) \left (\frac{\partial r}{\partial y}(\epsilon,x).h \right ) + h(- r(\epsilon,x)) \right ] $$
and 
$$ \frac{\partial}{\partial \epsilon} \left 
( \epsilon \frac{\partial g}{\partial y}(\epsilon,x)\cdot h  \right ) = 
 f^\prime(x(-r(\epsilon,x))) \left [ x^{\prime \prime}(- r(\epsilon,x)) \frac{\partial r}{\partial \epsilon}(\epsilon,x) 
\left ( \frac{\partial r }{\partial y}(\epsilon,x).h \right )\right ] $$ 
$$ + f^\prime(x(-r(\epsilon,x))) \left [
- x^\prime(-r(\epsilon,x))\frac{\partial^2 r}{\partial \epsilon \partial y}(\epsilon,x).h
- h^\prime(-r(\epsilon,x)) \frac{\partial r}{\partial \epsilon}(\epsilon,x) \right ]  $$
$$- \left [ -x^\prime(-r(\epsilon,x)) \frac{\partial r}{\partial y}(\epsilon,x).h + h(-r(\epsilon,x))  \right ] 
f^{\prime \prime}(x(-r(\epsilon,x))) x^\prime(-r(\epsilon,x))
\frac{\partial r}{\partial \epsilon}(\epsilon,x)$$ and
$$\frac{\partial}{\partial y} \left ( \epsilon \frac{\partial g}{ \partial y} \right )(\epsilon,x).(h,k) = $$
$$\left (-  \frac{\partial r}{\partial y}(\epsilon, x) .h + h(- r(\epsilon, x))  \right ) 
\left ( -\frac{\partial r}{\partial y}(\epsilon,x).k  + k(- r(\epsilon,x)) \right ) x^{\prime}(- r(\epsilon, x)) f^{\prime \prime}(x(-r(\epsilon,x))) $$
$$+ f^{\prime}(x(-r(\epsilon,x)))
\left [  
x^{\prime \prime}(- r(\epsilon,x)) \left ( \frac{\partial r}{\partial y}(\epsilon,x).k \right ) 
\left ( \frac{\partial r}{\partial y}(\epsilon,x).h \right )  \right ] $$
$$+ f^{\prime}(x(-r(\epsilon,x)))
\left [  
- x^{\prime}(-r(\epsilon,x)) \frac{\partial^{2} r}{\partial y^2}(\epsilon,x).(h,k)
 -h^{\prime}(-r(\epsilon,x)) \frac{\partial r}{\partial y}(\epsilon,x).k \right ] $$

If $r$ and $f$ are $C^2$, then one can check that $g : ]0,1[
\times\left ( U \cap C^2([-M,0], \mathbb{R}) \right ) \rightarrow
\mathbb{R} $ is $C^2$ and this is $H_4$.  If $\frac{\partial^2
  r}{\partial y^2} (\epsilon, x)$ has a continuous extension to
$C^{1}([-M,0],\mathbb{R}) \times C^{1}([-M,0],\mathbb{R}) \rightarrow
\mathbb{R}$, then so does $\frac{\partial^2 g}{\partial y^2}
(\epsilon, x)$ and this is $H_5$.  Since the extension of
$\frac{\partial r}{\partial y}(\epsilon, x)$ and $\frac{\partial^2
  r}{\partial y^2}(\epsilon, x)$ are continuous in $(\epsilon, x)$ the
first part of $H_6$ is satisfied.  Since $r : ]0,1[ \times U
\rightarrow \mathbb{R}$ is $C^1$, and due to the last hypothesis on
$\frac{\partial^2 r}{\partial y^2}$ in the proposition, the last
requirement of $H_6$ holds too.
\end{proof}

The hypotheses of theorem
 \ref{thm_hopf}: $f$ is $C^2(\mathbb{R}, \mathbb{R})$, $r : ]0,1[ \times
C^0([-M,0], \mathbb{R}) \rightarrow \mathbb{R}$ is $C^1$, and $r :
]0,1[ \times C^1([-M,0], \mathbb{R}) \rightarrow \mathbb{R}$ is $C^2$,
imply that  the regularity hypotheses of
Proposition \ref{prop1} and
Proposition \ref{prop2} are verified. \\

We now turn to the spectral hypotheses of Eichmann's Hopf-bifurcation
theorem.  We consider the equilibrium $x^{*}=0$, which satisfies
$g(\epsilon, 0) = 0$ for all $\epsilon \in ]0,1[$, and the linearized
equation at $x^*$ : $ y'(t) = \frac{\partial g}{\partial y}(\epsilon,
0) y_t \; \textnormal{ with } y_0 \in C^1([-M,0], \mathbb{R}) $, ie
$$ \epsilon y'(t) = - y(t) + f^\prime(0) y(t- r_0)  $$ where $r_0 = r(0)= r(x^*)=1$.
The characteristic equation is
$$ \left \{
\begin{array}{ccc}
1 + \epsilon \alpha &=& f^\prime(0) e^{-\alpha} \cos(\beta) \\
   \epsilon \beta &=& - f^\prime(0) e^{- \alpha} \sin{\beta},
 \end{array} \right . $$
This is the same characteristic equation for the constant-delay equation, and
a standard argument shows that there exists  a sequence 
$\epsilon_k \underset{k \rightarrow + \infty}{\longrightarrow} 0$ such that for each $\ep=\ep_k$ 
the characteristic equation has a single pair of solutions on the imaginary axis
$\lambda =\pm i \beta_k$, $\bt_k>0$.
This implies in particular that $(L_3)$ of \cite{eichmann} is satisfied. 

To check that for each $k$ the eigenvalue 
$\lambda_k$ can be tracked in a neighborhood of $\epsilon = \epsilon_k$, 
we use an implicit function theorem.
 Considering $G : ]0,1[ \times \mathbb{R} \times \mathbb{R}^2$ defined by 
$G(\epsilon, \alpha, \beta)=( 1 + \epsilon \alpha - f^\prime(0) e^{-\alpha} \cos(\beta)  ;
 \epsilon \beta + f^\prime(0) e^{-\alpha} \sin(\beta)   )  $, 
we have $G(\epsilon, \alpha, \beta)= 0$ if and only if 
$\lambda = \alpha + i \beta$ is a characteristic root. 
We have $G(\epsilon_k, 0, \beta_k)= 0$ (for any $k$), 
and the derivative of $G$ with respect to $\alpha$ and
$\beta$ at $(\epsilon_k, 0, \beta_k)$ is 
$$ \left (
\begin{array}{cc}
1 + \epsilon_k + f^\prime(0) \cos(\beta_k) &  f^\prime(0) \sin(\beta_k)   \\
- f^\prime(0)\sin(\beta_k) &  \epsilon_k + f^\prime(0) \cos(\beta_k)
\end{array} \right ) 
= \left ( \begin{array}{cc}
2 + \epsilon_k  &  f^\prime(0) \sin(\beta_k)   \\
- f^\prime(0)\sin(\beta_k) & 1+ \epsilon_k  
\end{array} \right )  $$ 
which has a positive determinant and is invertible. 
Thus (for any $k$) there is an open interval $I_k$ containing $\epsilon_k$ and 
a $C^1$ function
$\nu_k : I_k \rightarrow \mathbb{C}$ such that for
 all $\epsilon \in I_k$, $\nu_k(\epsilon)$ is 
a characteristic root and $\nu_k(\epsilon_k)= \lambda_k$. 

Furthermore, computing the $\epsilon$ derivative of $\nu_k$ at 
$\epsilon=\epsilon_k$ and  using the relations for 
$\lambda_k$, one finds 
$$\left \{ \begin{array}{ccc}
(1 + \epsilon_k ) \alpha^\prime(\epsilon_k) 
& = & \epsilon_k \beta_k \beta_k^\prime(\epsilon_k) \\
(1 + \epsilon_k ) \beta_k^{\prime} &=&
 - (1+ \epsilon_k \alpha_k^{\prime}(\epsilon_k)) \beta_k,
\end{array} \right . $$
which gives either $\beta_k^{\prime}(\epsilon_k)= 0$ 
and  $\alpha_k^{\prime}(\epsilon_k) = \frac{- 1}{\epsilon_k}$, or 
$\beta_k^{\prime}(\epsilon_k) \neq 0$ and
$ \alpha_k^{\prime}(\epsilon_k) = \frac{\epsilon_k \beta_k \beta_k^\prime(\epsilon_k)}{1 + \epsilon_k} $,
so that one always has $\alpha_k^\prime(\epsilon_k) \neq 0$. 
This means both $(L_1)$ and $(L_2)$ of Eichmann's spectral hypotheses
are satisfied,  which finishes  the proof of theorem \ref{thm_hopf}.

\section{A numerical method for solving transition layer equations.}
\label{app_tlo}

In this section we present a method for solving numerically
the transition layer equations (\ref{eqTLE}),
(\ref{eqTLEneg}), (\ref{eqTLEcst}) and (\ref{eqTLElambda}).
This section  is divided into two parts: negative and positive feedback.
Some  details are only provided for the positive feedback case, 
since for the negative feedback case they are similar.

\subsection{Positive feedback}
\label{sub_nummeth}

Existence of transition layer solutions is usually proven with the
help of a fixed point theorem.  Suppose that $f$ is smooth, has three
fixed points $x=-a$, $x=0$ and $x=b$, and that $f$ is increasing on
the interval $[-a,b]$ (positive feedback, with no symmetry
hypothesis).  Consider the set $\mathcal{C}^+$ of functions $\phi \in
C^1(\mathbb{R}, \mathbb{R})$ such that $\phi(0)=0$, $\underset{x
  \rightarrow - \infty}{\lim} \phi(x) = -a$, $\underset{x \rightarrow
  + \infty}{\lim} \phi(x) = b$, $\Vert \phi^\prime \Vert_{L^\infty}
\leq \underset{x \in [-a;b]}{ \sup} \vert f(x) \vert = \max \{ a, b
\}$, and $\phi$ is strictly increasing on $\mathbb{R}$.  On that space
we define the operator $\mathcal T$ that to $\phi \in \mathcal{C}^+$
associates the function $\mathcal T \phi = \psi$ given by
$$\psi(t) = e^{-t} \int_{- \infty}^{t} e^s f( \phi [ s + \rho  - R (\phi(s)) ] )ds,$$ 
where $\rho \in [0; + \infty[$ is the only constant such that 
$\psi(0) = \int_{-\infty}^{0} e^s f( \phi [ s + \rho - R(\phi(s))] )ds = 0 $.
It can be shown that $\psi\in\mathcal{C}^+$.
Then $\psi= \mathcal T \phi$ is the unique solution of equation 
\begin{equation}\label{eqpsi} \dot \psi(t)  = - \psi(t)
  + f( \phi [ t - R(\phi(t)) + \rho  ] ) \end{equation}
such that $\underset{x \rightarrow - \infty}{\lim} \psi(x) = -a$,
$\underset{x \rightarrow + \infty}{\lim} \psi(x) = b$ and $\psi(0)=0$.
So an increasing solution of the transition layer equation (\ref{eqTLE})
(case $\eta(0)=0$ and $\eta^\prime(0)=1$), 
is a fixed point of the operator $\mathcal T$ in the set $\mathcal{C}^+$.
In the following we show  how to  numerically solve
the fixed point problem 
$\mathcal T \phi^+ = \phi^+$. A solution to the problem 
$\mathcal T \phi^- = \phi^-$ is obtained in the same way.

Let $\phi$ be  a smooth increasing function  such that $\underset{t
  \rightarrow - \infty}{\lim} \phi(t) = -a $ and $\underset{t
  \rightarrow + \infty}{\lim} \phi(t) = b $. Then,   for any $\gamma \in
\mathbb R$, the function
$$\psi_\gamma (t) = e^{-t} \int_{- \infty}^{t} e^s f( \phi [ s + \gamma  - R (\phi(s)) ] )ds$$ 
is also smooth, increasing, and  satisfies $\underset{t
  \rightarrow - \infty}{\lim} \psi(t) = -a $ and $\underset{t
  \rightarrow + \infty}{\lim} \psi(t) = b $. Since, for any $t
\in \mathbb R$, the map $\gamma \mapsto \psi_\gamma(t)$ is increasing,
there is a unique $\gamma=\rho$ such that
$\psi_{\rho}(0)=0$ and so $\psi_{\rho} = \mathcal T \phi$.
A numerical approximation to the map $\phi\to \psi_\rho$ is the following.
Let  $L>0$ and $\phi^L$ be the function that coincides with 
$\phi$ on the interval $[-L,L]$ and such that $\phi^L(t)=-a$ for 
$t<-L$ and $\phi^L(t)=b$ for 
$t>L$. Clearly $\sup_{t\in\mathbb R}|\phi^L(t)-\phi(t)|$ can be made arbitrarily small
if $L$ is chosen sufficiently large. So, we fix $L>0$ and choose an initial 
$\phi_0$ that satisfies  $\phi_0(t)=-a$ for 
$t<-L$ and $\phi_0(t)=b$ for 
$t>L$. Then for a given $\gamma$ we use a first order Euler method to solve
the equation 
$$ \psi_\gamma^\prime(t) = - \psi_\gamma(t) + 
f \left ( \phi \left [ t + \gamma - R(\phi_0(t)) \right ]  \right ), $$
with the initial condition $\psi_\gamma(-L) = -a$ and time step $dt$.
Knowing that $\gamma > \rho$ is equivalent to $\psi_\gamma(0) >0$, 
we can use a shooting method to find $\gamma_1$ and $\gamma_2$   
such that $ \rho - dt \leq \gamma_1 \leq \rho \leq \gamma_2 \leq \rho + dt $.
Defining $\phi_1(t)=\psi_{\gamma_1}(t)$ for $t\le L$ and $\phi_1(t)=b$ for $x>L$ 
we obtain the  approximation  $\mathcal T \phi_0 \approx  \phi_1$.
This procedure can be iterated $\phi_{n+1} = \mathcal T \phi_n$, $n=1,2,\ldots$ 
hoping  that it converges to a fixed point $\phi^+$.
The convergence of this sequence was  numerically verified, typically 
iterating $\phi_{n+1} = \mathcal T \phi_n$ up to $n=20$ and checking that 
$$
 \vert \phi_n(t) - \phi_{n-1}(t) \vert \leq dt.
$$
A sample of our  results are shown  
in table \ref{Fig_tableau_rho} (a)  (section \ref{subsub_pf}) and  in
figure \ref{Fig_rho_1}.

In table \ref{Fig_tableau_rho} (a) the constants $\rho^\pm$ were
computed for  $f(x)=\frac12 \arctan(5x)$, $L=100.0$, and $dt=0.001$.
The initial condition used to obtain table \ref{Fig_tableau_rho} (a)
was  $\phi_0^+(t)= \frac{ab(1-e^{-\frac23 t})}{a+b e^{-\frac23 t}} $
 for the increasing transition layer solution, and 
$\phi^-_0(t) = \frac{ab(1-e^{\frac23 t})}{a+b e^{\frac23 t}}$ 
for the decreasing transition layer solutions.

In figure \ref{Fig_rho_1} (a) several iterates $\phi_n$ are shown to
converge to a limit profile $\phi^-$ for $f(x)=\frac12 \arctan(5x)$,
with $R(x)=cos(x)$, and initial condition $\phi_0(t) = ab \frac{1 -
  e^{t}}{a + be^t}$. 
In that case the feedback $f$ is symmetric, 
so that the operator $\mathcal T$ is also symmetric, 
and if we consider an increasing initial profile $\tilde \phi_0 = - \phi_0$, 
the corresponding sequence is $\tilde \phi_n = \mathcal T^n (-\phi_0) = - \mathcal T \phi_0$. 
In particular, figure \ref{Fig_rho_1} illustrates the convergence and shapes
 of iterates for both increasing and decreasing profiles. 
 In figure \ref{Fig_rho_1} (b) the limit profile
$ \phi_{ 10 } (t) = \mathcal T^{ 10 } \phi_0 (t) \approx \phi^-(t)$ is
shown for $R(x)= 2x$, $R(x)= -x^2$, and $R(x)= \frac14 x(1+x)$, using 
the initial condition $\phi_0(t) = ab \frac{1 - e^{t}}{a + be^t}$.
For $R(x)=x$, the numerical convergence of the sequence $\phi_{n+1} =
\mathcal T \phi_n$ was successfully tested for the following initial
functions:
$\phi_0(t) = - \frac{t}{ \vert t \vert}$, and 
$\phi_0(t) = ab \frac{1 - e^{t}}{a + be^t} (1+ 0.4 \cos(t))$. 
All these  results also hold for  a non-symmetric positive feedback as well.

\begin{figure}[th!]
\begin{center}
\includegraphics[width=7.5cm]{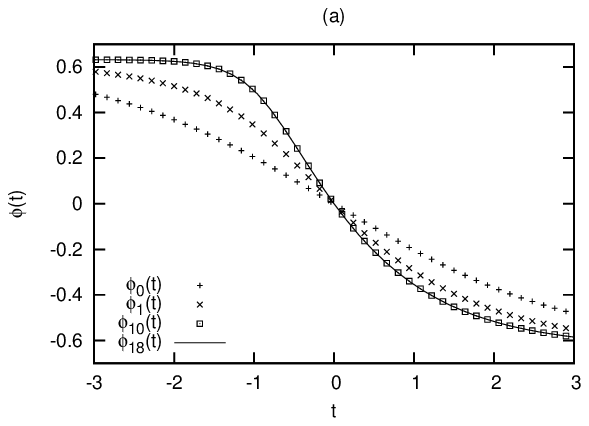}
\includegraphics[width=7.5cm]{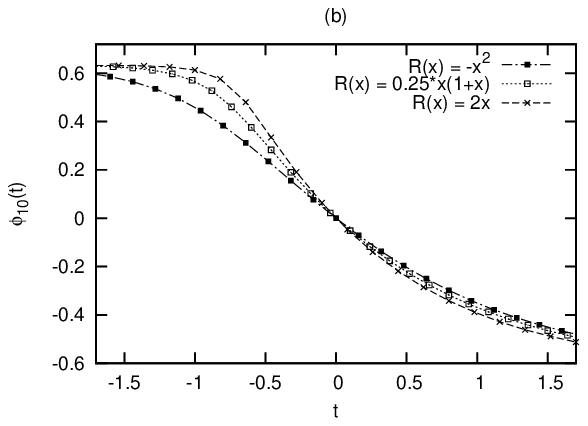}
\end{center}
\vspace{-1.0cm}
\caption{\small{Decreasing transition layer solutions of equation (\ref{eqTLE})
($\eta(0)=0$ with $\eta^\prime(0)=1$) 
with  $f(x)=\frac12 \arctan(5x)$
(Positive feedback case). 
 In (a): $R(x)=cos(x)$  and $\phi_0(t) = ab \frac{1 - e^{t}}{a + be^t}$. 
In (c): $ \phi_{ 10 } (t) = \mathcal T^{ 10 } \phi_0 (t)$ is displayed 
(with $\phi_0(t) = ab \frac{1 - e^{t}}{a + be^t}$), 
for  functions $R(x)= 2x$, $R(x)= -x^2$, and $R(x)= \frac14 x(1+x)$.  
}}
\label{Fig_rho_1}
\end{figure}

For DDEs with state dependent delays, such as transition layer equations (\ref{eqTLE}), 
the existence of solutions $\phi$ to the Cauchy problem is known under classical hypotheses, 
in particular on the delay function $R$, 
that ensures that $t \mapsto t - R(\phi(t))$ is not decreasing 
(see \cite{Hartung_2006} for example).
We mention that in several numerical examples, when
$R$ is too large, the maps $t \mapsto t - R(\phi_n(t))$ are not monotone, 
the iterative sequences $\phi_n=\mathcal T^n \phi_0$ does not converge, 
and transition layer solutions seem not to exist (see section \ref{subsection_tle_alpha>1})

 In the case $\eta(0)=\eta^\prime(0)=0$, the transition layer equation
 (\ref{eqTLElambda}) is associated to the operator 
$$\mathcal T_\lambda \phi(t) = e^{-t} \int_{- \infty}^{t} e^s f( \phi [ s + \rho^+  - \lambda R (\phi(s)) ] )ds,  $$ 
and the same numerical method presented above can be used to compute
its fixed point.  For instance, figure \ref{Fig_rho_pm_alpha>1}
(section \ref{subsubsection_tle_alpha>1-pos}) shows the constants
$\rho^\pm$ that were computed using various choices of $R$ and:
$f(x)=\frac12 \arctan(5x)$, $L=100.0$, $dt=0.001$, $n=20$, and the
initial condition $\phi_0^+(t)=\frac{ab(1-e^{-\frac23 t})}{a+b
  e^{-\frac23 t}}$ for the increasing transition layer solution and
$\phi^-_0(t) = \frac{ab(1-e^{\frac23 t})}{a+b e^{\frac23 t}}$ for the
decreasing transition layer.

Finally, in figure \ref{Fig_rho_3} (in section \ref{subsec_meta}),
illustrating that metastability may be induced by state dependent
delay in the case of non-symmetric positive feedback (that does not
exhibit metastability for constant delay),
the values of $\rho^\pm$ were computed using $L=100.0$, $dt=0.001$,
$n=20$, and the initial conditions $\phi_0^+(t)=
\frac{ab(1-e^{-\frac23 t})}{a+b e^{-\frac23 t}}$, $\phi^-_0(t) =
\frac{ab(1-e^{\frac23 t})}{a+b e^{\frac23 t}}$.

\subsection{Negative feedback}
\label{sub_nummeth_neg}

For a negative feedback $f$, one can repeat the procedure above 
and define  the  operator $\mathcal T$ 
that has as fixed points the solutions to the  transition layer equation
 (\ref{eqTLEneg}), namely 
\begin{equation*}
\begin{array}{rcl}
\phi^+ (t)  &=& \mathcal T \phi^- (t) = e^{-t} \int_{- \infty}^{t} e^s f( \phi^- [ s + \rho^+  - R (\phi^-(s)) ] )ds \\
\phi^- (t)  &=& \mathcal T \phi^+ (t) = e^{-t} \int_{- \infty}^{t} e^s f( \phi^+ [ s + \rho^-  - R (\phi^+(s)) ] )ds 
\end{array}
\end{equation*} 
where $\phi^+$ (resp. $\phi^-$) are increasing (resp. decreasing)
smooth functions with $\underset{t \rightarrow - \infty}{\lim} \phi^+
= -a $, $\underset{t \rightarrow + \infty}{\lim} \phi^+ = b $ and
$\underset{t \rightarrow - \infty}{\lim} \phi^- = b $, $\underset{t
  \rightarrow + \infty}{\lim} \phi^- = -a$, and $\rho^\pm$ are the
only constants such that $\phi^+(0)=\phi^-(0)=0$.  Since the function
$f$ is decreasing, the operator $\mathcal T$ now maps an increasing
function to a decreasing one and vice versa.  As in the case of
positive feedback, if the function $\phi$ is strictly monotone the
function $\gamma \mapsto \Psi_\gamma (t)$ is strictly monotone too, so
that the constants $\rho^\pm$ above are well defined and can be
computed by a shooting method.  We choose a smooth increasing initial
function $\phi_0$ with $\underset{t \rightarrow - \infty}{\lim} \phi^+
= -a $, $\underset{t \rightarrow + \infty}{\lim} \phi^+ = b $ and the
subsequences $\phi_{2n}$ and $\rho_{2n}$ (resp. $\phi_{2n+1}$ and
$\rho_{2n+1}$) converge to the transition layer solution $\phi^+$ and
the constant $\rho^+$ (resp. $\phi^-$ and $\rho^-$ ).  Numerically, as
in the positive feedback case, we use discretization of step $dt$ on
an interval $[-L, L]$, the functions $\Psi_\gamma$ are computed using
a first order Euler scheme of step $dt$, the constant $\rho$ are
approximated with a precision $dt$, and the convergence after $n$
iterations of $\mathcal T$ is checked similarly :
\begin{eqnarray*}
 \vert \phi_n(t) - \phi_{n-2}(t) \vert &\leq& dt,\\
\vert \phi_{I-1}(t) - \phi_{I-3}(t) \vert &\leq& dt,
\end{eqnarray*}
In table \ref{Fig_tableau_rho_NF} (section \ref{subsub_nf})
 the constants $\rho^\pm$ were computed using $L=100.0$, $dt=0.001$,  $n=40$, 
and  $\phi^+_0(t) = \frac{ab(1-e^{-\frac23 t})}{a+b e^{-\frac23 t}} $.

\medskip 
\medskip 
{\bf  ACKNOWLEDGMENTS}\\
\noindent
The authors thank Denis Mestivier for his help in handling RADAR-V codes.


\begin{thebibliography}{99}

%
\bibitem{alt}
W. Alt (1978): Some periodicity criteria for functional differential 
equations, {\it Manuscripta Math.} {\bf 23},
295-318.

%
\bibitem{arino_benkhalti} O. Arino and R. Benkhalti (1988) Periodic
  solutions for: $x(t)=\lambda f(x(t),x(t-1))$, {\it Proc. Roy. Soc.
    Edinburgh Sect. A}, {\bf 109}, 245-260.
%
\bibitem{arino-hbid}
O. Arino, K. P. Hadeler, and M. L. Hbid (1998):
 Existence of periodic solutions for delay
differential equations with state dependent delay, 
{\it J. Diff. Eq.}{\bf 144},263-301.

%
\bibitem{arino_seguier} O. Arino, and P. S\'eguier (1979) Existence of
  oscillating solutions for certain differential equations with delay.
  Functional Differential Equations and Approximation of Fixed Point
  (H.-O. Peitgen, H.-O. Walthers (Eds)) Lecture Notes in Mathematics
  430: 46-64 Springer Verlag, New York.

%
\bibitem{OA-ES-AF_99} O. Arino, E. S\`anchez, A. Fathallah (2001):
  State-dependent delay differential equations in population dynamics:
  Modeling and analysis.  Topics in Functional Differential and
  Difference Equations (T. Faria, P. Freitas (eds)) Fields Inst.
  Commun. 29 : 19-36, A.M.S., Providence RI.

%
\bibitem{bartha_monotone} M. Bartha (2001): 
Convergence of Solutions for an Equation
with State-Dependent Delay, {\it J. Math. Analysis and Applic.} {\bf 254},
 410-432.

%
\bibitem{bartha_periodic} M. Bartha (2003):
 Periodic solutions for differential equations with
state-dependent delay and positive feedback, 
{\it Nonlinear Anal.} {\bf 53}, 839-857.

%
\bibitem{JB-MM_89} J. B\'elair, M. Mackey (1989): Consumer memory and
  price fluctuations in commodity markets: An integrodifferential
  model {\it Journal of Dynamics and Differential Equations} {\bf 1 },
  299-325.

%
\bibitem{cp} J. Carr and R. L. Pego (1989):
Metastable patterns in solutions of $u_t=\ep^2u_{xx}-f(u)$,
{\it Comm. Pure Appl. Math. \bf XLII}, 523-576.

%
\bibitem{chowmp} S.-N. Chow and  J. Mallet-Paret (1983):
 Singularly perturbed delay-differential equations.
 North-Holland Math. Stud. 80, 7–12.

%
\bibitem{cooke-huang} K.L. Cooke and W. Huang (1996):
On the problem of linearization for state-dependent delay differential
equations, {\it Proc. AMS} {\bf 124}, 1417-1426.

%
\bibitem{eichmann} M. Eichmann (2006): 
A local Hopf Bifurcation Theorem for differential
equations with state - dependent delays, PhD Thesis,
Department of Mathematics,
Justus - Liebig - University Giessen, Giessen-Germany.

%
\bibitem{CF-MM_09} C. Foley, M. Mackey (2009):
 Dynamic hematological disease: a review,
{\it Journal of Mathematical Biology} {\bf 58},  285-322.

%
\bibitem{fh} G. Fusco and J. K. Hale (1989): Slow-motion
manifolds, dormant instability, and singular perturbations,
{\it J. Dyn. Diff. Eq. \bf 1}, 75-94.

%
\bibitem{politi2} 
G. Giacomelli, R. Meucci, A. Politi,  and F. T. Arecchi (1994): 
Defects and Spacelike Properties of Delayed Dynamical Systems,
{\it Phys. Rev. Lett.} {\bf  73}, 1099–1102.

%
\bibitem{politi}
G. Giacomelli and A. Politi (1998):
 Multiple scale analysis of delayed dynamical systems,
{\it Physica D} {\bf 117}, 26–42.

%
\bibitem{pre} C. Grotta-Ragazzo, K. Pakdaman, and C. P. Malta
(1999): Metastability for delayed differential equations,
{\it Phys. Rev. E \bf 60}, 6230-6233.

%
\bibitem{jdde} C. Grotta-Ragazzo, C. P. Malta, and K. Pakdaman (2010):
Metastable Periodic Patterns in Singularly Perturbed
Delayed Equations
{\it J. Dyn. Diff. Eq.  \bf 22}, 203-252.

%
\bibitem{Hartung_2006} F. Hartung, T. Krisztin, W. Hans-Otto, J. Wu (2006):
Functional differential equations with state-dependent delays: theory and applications. 
Handbook of differential equations: ordinary differential equations. Vol. III, 435-545,
Handb. Differ. Equ., Elsevier/North-Holland, Amsterdam, 2006.     

%
\bibitem{krisztin-arino} T. Krisztin, O. Arino (2001):
 The 2-dimensional attractor of a differential 
equation with state-dependent delay,
{\it J. Dyn. Diff. Eq.} {\bf  13}, 453–522.

%
\bibitem{kris} T. Krisztin, H.-O. Walther, and J. Wu (1999):
{\it Smoothness and Invariant Stratification of an Attracting Set
for Delayed Monotone Positive Feedback}, Fields Institute Monograph
Series, AMS, Providence-RI.

%
\bibitem{smith} Y. Kuang and  H.L. Smith (1992):
 Slowly oscillating periodic solutions of autonomous state-dependent delay
differential equations, {\it Nonlinear Anal.} {\bf 19}, 855–872.

%
\bibitem{AL-JM_89} A. Longtin, J. Milton (1989): Modelling autonomous
  oscillations in the human pupil light reflex using non-linear
  delay-differential equations {\it Bulletin of Mathematical Biology}
  {\bf 51 }, 605-624.

%
\bibitem{AL-JM_89B} A. Longtin, J. Milton (1989): Insight into the
  transfer function, gain, and oscillation onset for the pupil light
  reflex using nonlinear delay-differential equations {\it Biological
    Cybernetics } {\bf 61 }, 51-58.

%
\bibitem{MM_89} M. Mackey (1989): Commodity price fluctuations: Price
  dependent delays and nonlinearities as explanatory factors {\it
    Journal of Economic Theory} {\bf 48 }, 497 - 509.

%
\bibitem{mallet1} J. Mallet-Paret (1988):  Morse decompositions
for differential delay equations, {\it J. Diff. Eq}. {\bf 72},
270-315.

%
\bibitem{m-pn} J. Mallet-Paret and R. Nussbaum (1986):
 Global continuation and asymptotic behavior for periodic solutions
of a diff-delay equation, 
{\it Annali di Matematica Pura ed Applicada} (4) {\bf CXLV}, 33-128.

%
\bibitem{MPN1} J. Mallet-Paret and R. Nussbaum (1992):
Boundary layer phenomena for differential-delay equations with
state-dependent
time lags, I; {\it Arch. Rat. Mech. Anal. \bf 120}, 99-146.

%
\bibitem{MPN2} J. Mallet-Paret and R. Nussbaum (1996):
Boundary layer phenomena for differential-delay equations with
state-dependent
time lags: II, {\it J. Reine Angew. Math. \bf 477}, 129-198.

%
\bibitem{MPN3} J. Mallet-Paret and R. Nussbaum (2003):
Boundary layer phenomena for differential-delay equations with
state-dependent
time lags: III, {\it J. Diff. Eq. \bf 189}, 640-692.

%
\bibitem{MPN4} J. Mallet-Paret and R. Nussbaum (2011):
Superstability and rigorous asymptotics in singularly 
perturbed state-dependent delay-differential equations, 
{\it J. Diff. Eq. \bf 250}, 4037-4084.

%
\bibitem{MPN5} J. Mallet-Paret and R. Nussbaum (2011): Stability of
  periodic solutions of state-dependent delay-differential equations,
  {\it J. Diff. Eq. \bf 250}, 4085-4103.

%
\bibitem{MPNP} J. Mallet-Paret, R.D. Nussbaum, and  P. Paraskevopoulos
(1994): Periodic solutions for functional differential
equations with multiple state-dependent time lags, 
{\it Topol. Methods Nonlinear Anal.} {\bf  3},
101-162.

%
\bibitem{Milton_2010} J. Milton, P. Naik, C. Chan and S. A. Campbell  (2010): 
Indecision in neural decision making models,
{\it Mathematical Modeling of Natural Phenomena} {\bf 5}, 125-145.

%
\bibitem{Milton_2011} J. Milton, A. Quan and I. Osorio . (2011): 
Nocturnal frontal lobe epilepsy: Metastability in a dynamic disease ? 
The intersection of neurosciences, biology, mathematics, engineering and physics 
I. Osorio, H. P. Zavari, M. G. Frei and S. Arthurs, editors CRC Press, Boca Raton : 501-510. 

%
\bibitem{niz1} M. Nizette (2004): Stability of square oscillations
in a delayed-feedback system, {\it Phys. Rev. E} {\bf 70},
p. 056204-1 to 056204-6.

%
\bibitem{niz2}  M. Nizette (2003): Front dynamics in a delayed-feedback
 system  with external forcing, {\it Physica D} {\bf 183}, 220-244.

%
\bibitem{Nussbaum_2003} R. D. Nussbaum (2003):  Limiting profiles for solutions of differential-delay equations,
Dynamical Systems, Lecture Notes in Mathematics, Volume 1822/2003, 299-342,
Springer Berlin / Heidelberg, 2003.

%
\bibitem{pre_ring} K. Pakdaman, C.P. Malta, C. Grotta-Ragazzo, O. Arino
and J.-F. Vibert (1997):
Transient oscillations in continuous-time excitatory ring neural networks
with delay,
{\it Phys. Rev. E} {\bf 55}, 3234–3248.

%
\bibitem{pre_expo}
K. Pakdaman, C. Grotta-Ragazzo and C.P. Malta (1998):
Transient regime duration in continuous-time neural networks with delay,
{\it Phys. Rev. E} {\bf 58}, 3623–3627.

%
\bibitem{polner} M. Polner (2002): Morse decomposition
for delay-differential  equations with positive feedback,
{\it Nonlinear Analysis} {\bf 48}, 377-397.

%
\bibitem{shar} A. N. Sharkovsky, Yu. L. Maistrenko,  E. Yu. Romanenko (1993):
 Difference Equations and Their Applications. Kluwer, Dordrecht.

%
\bibitem{walther} H.-O. Walther (2002):
 Stable periodic motion of a system with state dependent delay, 
{\it Diff. and  Integral
Eq.} {\bf 15},  923–944.


\end{thebibliography}
\end{document}